\documentclass[envcountsame]{llncs}
\usepackage{amsmath,amssymb}
\usepackage{tabularx,graphicx}
\usepackage{xspace}
\usepackage{mathtools}
\usepackage{mathptmx}
\usepackage{rotating}
\usepackage{xspace}

\usepackage{tikz}
\usepackage{url}
\usetikzlibrary{arrows,shapes}

\DeclareMathOperator\optsol{Optsol}

\DeclareMathOperator\arity{ar}

\DeclareMathOperator*\argmin{\arg\,\min}
\DeclareMathOperator*\argmax{\arg\,\max}
\DeclareMathOperator\var{Var}
\DeclareMathOperator\val{val}

\newcommand{\reduces}{\ensuremath{\leq^{\mathrm{CV}}}}

\newcommand{\Nand}{\ensuremath{\textsc{NAND}}}

\newcommand{\overbar}[1]{\mkern 2.7mu\overline{\mkern-2.7mu#1\mkern-2.7mu}\mkern 2.7mu}
\newcommand{\clone}[1]{\ensuremath{\mathcal{#1}}}
\newcommand{\cclone}[1]{\ensuremath{\langle #1 \rangle}}
\newcommand{\pcclone}[1]{\ensuremath{\langle #1 \rangle_{\nexists}}}

\newcommand{\eq}{{\rm eq}}
\newcommand{\neqq}{{\rm neq}}
\newcommand{\ppol}{\ensuremath{\rm{pPol}}}
\newcommand{\pol}{\ensuremath{\rm{Pol}}}

\newcommand{\F}{{\rm F}} % unary constraint {0}
\newcommand{\T}{{\rm T}} % unary constraint {1}

\newcommand{\Rddd}{\ensuremath{R_{\scriptscriptstyle 3
      \neq}^{\scriptscriptstyle 1/3}}}

\newcommand{\orn}[2]{\ensuremath{\mathrm{OR}^{
      #2}_{\scriptscriptstyle #1}}}
\newcommand{\nandn}[2]{\ensuremath{\mathrm{NAND}^{#2}_{\scriptscriptstyle
      #1}}}
\newcommand{\oddn}[2]{\ensuremath{\mathrm{ODD}^{#2}_{\scriptscriptstyle
      #1}}}
\newcommand{\evenn}[2]{\ensuremath{\mathrm{EVEN}^{#2}_{\scriptscriptstyle
      #1}}}
\newcommand{\cc}[3]{\ensuremath{\cloneFont{{#1}^{#2}_{#3}}}}

\newcommand{\cloneFont}[1]{\mathsf{#1}}
\newcommand{\problemFont}[1]{\textsc{#1}}
\newcommand{\mathCommandFont}[1]{\mathrm{#1}}

\newcommand{\tup}[1]{\mathbf{#1}}

\newcommand{\Wprefix}{\protect\ensuremath\problemFont{w}}
\newcommand{\UWprefix}{\protect\ensuremath\problemFont{u}}

\newcommand{\CSP}{\protect\ensuremath\problemFont{CSP}}
\newcommand{\VCSP}{\protect\ensuremath\problemFont{VCSP}}
\newcommand{\MAXCUT}{\protect\ensuremath\problemFont{Max-Cut}}
\newcommand{\UWVCSP}{\UWprefix\protect\ensuremath\problemFont{-VCSP}}
\newcommand{\MO}{\UWprefix\protect\ensuremath\problemFont{-Max-Ones}}
\newcommand{\Maxones}{\protect\ensuremath\problemFont{Max-Ones}}
\newcommand{\WMO}{\Wprefix\protect\ensuremath\problemFont{-Max-Ones}}
\newcommand{\MinO}{\UWprefix\protect\ensuremath\problemFont{-Min-Ones}}

\newcommand{\MAXCSP}{\protect\ensuremath\problemFont{Max-CSP}}
\newcommand{\MINCSP}{\protect\ensuremath\problemFont{Min-CSP}}

\newcommand{\SAT}{\protect\ensuremath\problemFont{SAT}}

\newcommand{\CloneBF}{\protect\ensuremath{\cloneFont{BF}}}
\newcommand{\CloneM}{\protect\ensuremath{\cloneFont{M}}}
\newcommand{\CloneL}{\protect\ensuremath{\cloneFont{L}}}
\newcommand{\CloneR}{\protect\ensuremath{\cloneFont{R}}}

\newcommand{\CloneD}{\protect\ensuremath{\cloneFont{D}}}
\newcommand{\CloneN}{\protect\ensuremath{\cloneFont{N}}}

\newcommand{\CloneS}{\protect\ensuremath{\cloneFont{S}}}

\newcommand{\CloneV}{\protect\ensuremath{\cloneFont{V}}}
\newcommand{\CloneE}{\protect\ensuremath{\cloneFont{E}}}
\newcommand{\CloneI}{\protect\ensuremath{\cloneFont{I}}}
\newcommand{\CoCloneII}{\protect\ensuremath{\cloneFont{II}}}

\newcommand{\B}{\mathbb{B}}

\newcommand{\SE}{\textsc{SE}}

\newcommand{\Rel}[1]{\ensuremath{R_{\cloneFont{#1}}}}
\newcommand{\id}{{\protect\ensuremath{\mathCommandFont{id}}}}
\newcommand{\dual}{{\protect\ensuremath{\mathCommandFont{dual}}}}

\begin{document}
\title{Relating the Time Complexity of Optimization
  Problems in Light of the Exponential-Time Hypothesis\thanks{This is
    an extended version of {\em Relating the Time Complexity of
      Optimization Problems in Light of the Exponential-Time
      Hypothesis}, appearing in Proceedings of the 39th International
    Symposium on Mathematical Foundations of Computer Science MFCS
    2014 Budapest, August 25-29, 2014}}

\author{Peter Jonsson\inst{1}, Victor Lagerkvist\inst{1},
Johannes Schmidt\inst{1} and Hannes Uppman\inst{1}} 

\institute{Department of Computer and Information Science, Link\"oping University, Sweden\\
\email{\{peter.jonsson, victor.lagerkvist, johannes.schmidt, hannes.uppman\}@liu.se}}
\maketitle

\begin{abstract} Obtaining lower bounds for NP-hard problems has for a long
  time been an active area of research. Recent algebraic techniques introduced by
  Jonsson et al.$\,$(SODA 2013) show that the time complexity of the
  parameterized $\SAT(\cdot)$ problem correlates to the lattice of
  strong partial clones. With this ordering they
  isolated a relation $R$ such that $\SAT(R)$ can
  be solved at least as fast as any other NP-hard $\SAT(\cdot)$
  problem. In this paper we extend this method and show that such
  languages also exist for the {\em max ones problem} ($\Maxones(\Gamma)$)
  and the {\em Boolean valued constraint satisfaction problem} over
  finite-valued constraint languages ($\VCSP(\Delta)$).
%These languages may be interesting when investigating the borderline
%between polynomial time, subexponential time and exponential time
%algorithms since they in a precise sense can be regarded as
%NP-hard problems with minimum time complexity. 
With the help of these languages we relate $\Maxones$ and $\VCSP$ to the exponential time hypothesis
in several different ways.
%, and show, for instance, that if
%$\VCSP(\Delta)$ is NP-hard and subexponential for some $\Delta$, then
%(unweighted) $\Maxones(\Gamma)$, and $\SAT(\Gamma)$ are solvable in
%subexponential time for every choice of $\Gamma$ and, consequently,
%that the ETH is false.
% The proofs
%  make heavy use the algebraic approach to the complexity of constraint
%  satisfaction problems and in particular employ the {\em weak base
%    method} (by Schnoor and Schnoor) for the $\Maxones$
%  problems and {\em multimorphisms} (by Cohen et al.) for $\VCSP$.
\end{abstract}

\section{Introduction}
A superficial analysis of the NP-complete problems may lead one to
think that they are a highly uniform class of problems: in fact, under
polynomial-time reductions, the NP-complete problems may be viewed as
a {\em single} problem.  However, there are many indications (both
from practical and theoretical viewpoints) that the NP-complete
problems are a diverse set of problems with highly varying properties,
and this becomes visible as soon as one starts using more refined
methods.  This has inspired a strong line of research on the ``inner
structure'' of the set of NP-complete problem. Examples include the
intensive search for faster algorithms for NP-complete
problems~\cite{Woeginger:CO2003} and the highly influential work on
the {\em exponential time hypothesis} (ETH) and its
variants~\cite{Lokshtanov:etal:beatcs2011}. Such research might not
directly resolve whether P is equal to NP or not, but rather attempts
to explain the seemingly large difference in complexity between
NP-hard problems and what makes one problem harder than
another. Unfortunately there is still a lack of general methods for
studying and comparing the complexity of NP-complete problems with
more restricted notions of reducibility. Jonsson et
al.~\cite{Jonsson:etal:soda2013} presented a framework based on
{\em clone theory}, applicable to problems that can
be viewed as ``assigning values to variables'', such as constraint
satisfaction problems, the vertex cover problem, and integer
programming problems. To analyze and relate the complexity of these problems in
greater detail we utilize
%{\em LV-} and {\em CV-reductions} --- 
%and is helpful when
%analyzing problems with {\em LV-} and {\em CV-reductions}. 
%these are
polynomial-time reductions which increase the number of variables by a
constant factor ({\em linear variable reductions} or {\em LV-reductions})
and reductions which increases the amount of variables by a constant
({\em constant variable reductions} or {\em CV-reductions}). Note the
following: (1) if a problem $A$ is solvable in $O(c^n)$ time (where
$n$ denotes the number of variables) for all $c > 1$ and if problem
$B$ is LV-reducible to $A$ then $B$ is also solvable in $O(c^n)$ time
for all $c > 1$ and (2) if $A$ is solvable in time $O(c^n)$ and if $B$
is CV-reducible to $A$ then $B$ is also solvable in time
$O(c^n)$. Thus LV-reductions preserve subexponential complexity while
CV-reductions preserve exact complexity.  Jonsson et
al.~\cite{Jonsson:etal:soda2013} exclusively studied the Boolean
satisfiability $\SAT(\cdot)$ problem and identified an NP-hard
$\SAT(\{R\})$ problem CV-reducible to all other NP-hard
$\SAT(\cdot)$ problems. Hence $\SAT(\{R\})$ is, in a sense, the {\em
  easiest} NP-complete $\SAT(\cdot)$ problem since if $\SAT(\Gamma)$ can be
solved in $O(c^n)$ time, then this holds for $\SAT(\{R\})$, too.  With
the aid of this result, they analyzed the consequences of
subexponentially solvable $\SAT(\cdot)$ problems by utilizing the
interplay between CV- and LV-reductions.  As a by-product, Santhanam
and Srinivasan's~\cite{Santhanam:Srinivasan:icalp2012} negative result
on sparsification of infinite constraint languages was shown not to
hold for finite languages.

We believe that the existence and construction of such easiest
languages forms an important puzzle piece in the quest of relating the
complexity of NP-hard problems with each other, since it effectively
gives a lower bound on the time complexity of a given problem with
respect to constraint language restrictions. As a logical continuation
on the work on $\SAT(\cdot)$ we pursue the study of CV- and
LV-reducibility in the context of Boolean optimization problems. In
particular we investigate the complexity of $\Maxones(\cdot)$ and
$\VCSP(\cdot)$ and introduce and extend several non-trivial methods
for this purpose. The results confirms that methods based on universal
algebra are indeed useful when studying broader classes of NP-complete
problems.
%{\em surjective CSP} ($\SURCSP$), {\em max ones}
%($\MO$) and {\em valued CSP} ($\VCSP$)
%Analyzing these problems is not straightforward and we need to
%introduce new techniques in both cases.
%The $\CSPSTAR(\Gamma)$ problem is the $\SAT(\Gamma)$ problem where the satisfying solution must also be surjective. 
%{\bf Write something why this is an interesting problem.}
%
The $\Maxones(\cdot)$ problem~\cite{Khanna:etal:sicomp2000} is a variant of $\SAT(\cdot)$
%problem
where the goal is to find a satisfying assignment which
maximizes the number of variables assigned the value 1. This problem is
closely related to the \textsc{0/1 Linear Programming} problem. The
$\VCSP(\cdot)$ problem is a function minimization problem that generalizes the
$\problemFont{Max-CSP}$ and $\problemFont{Min-CSP}$ problems~\cite{Khanna:etal:sicomp2000}.
We treat both the unweighted and
weighted versions of these problems and use the prefix $\UWprefix$ to denote
the unweighted problem and $\Wprefix$ to denote the weighted
version. 
%For an overview of Boolean optimization problems, the reader is referred to Creignou et
%al~\cite{cks01}. 
These problems are well-studied with respect to separating tractable
cases from NP-hard cases~\cite{Khanna:etal:sicomp2000,thapper2013} but
much less is known when considering the weaker schemes of
LV-reductions and CV-reductions. We begin (in
Section~\ref{section:maxones}) by identifying the easiest language for
$\WMO(\cdot)$. The proofs make heavy use of the {\em algebraic method} for
constraint satisfaction
problems~\cite{jeavons1998,Jeavons:etal:closure} and the {\em weak
  base method}~\cite{schnoor2008b}. The algebraic method
was introduced for studying the computational complexity of constraint
satsifaction problems up to polynomial-time reductions while the weak
base method~\cite{schnoor2008a} was shown by Jonsson et
al.~\cite{Jonsson:etal:soda2013} to be useful for studying
CV-reductions. To prove the main result we
however need even more powerful reduction techniques based on {\em weighted
  primitive positive implementations}~\cite{fourel,thapper}. For $\VCSP(\cdot)$ the
situation differs even more since the algebraic techniques developed for
$\CSP(\cdot)$ are not applicable --- instead we use 
\emph{multimorphisms}~\cite{vcsp22} when considering the
complexity of $\VCSP(\cdot)$. 
We prove (in Section ~\ref{section:vcsp}) that the
binary function $f_{\ne}$ which returns $0$ if its two arguments are
different and $1$ otherwise,
%$f_{\ne}(x,y) = 1-|x-y|$
results in the easiest NP-hard $\VCSP(\cdot)$ problem. This problem is
very familiar since it is the {\sc Max Cut} problem slightly
disguised.  The complexity landscape surrounding these problems is
outlined in Section~\ref{sec:broadpicture}.
 
With the aid of the languages identified in Section~\ref{section:easy_problems},
we continue (in Section~\ref{section:eth}) by relating $\Maxones$ and
$\VCSP$ with LV-reductions and connect them with the ETH. Our results
imply that (1) if the ETH is true then no NP-complete $\MO(\Gamma)$, $\WMO(\Gamma)$,
or $\VCSP(\Delta)$ is solvable in subexponential time and (2) that if the ETH is false then $\MO(\Gamma)$ and
$\UWVCSP_d(\Delta)$ are solvable in subexponential time for
%all
every
choice of $\Gamma$ and $\Delta$ and $d
\geq 0$. Here $\UWVCSP_d(\Delta)$ is the $\UWVCSP(\Delta)$ problem restricted to instances
%with a linear number of constraints with respect to $d$.
where the sum to minimize contains at most $dn$ terms.
%with at most $dn$ terms in the sum defining the objective function (where $n$ is the number of variables).
Thus, to disprove the ETH, our result implies
that it is sufficient to find a
single language $\Gamma$ or a set of cost functions $\Delta$ such that
 $\MO(\Gamma)$, $\WMO(\Gamma)$
%, $\UWVCSP_d(\Delta)$
or $\VCSP(\Delta)$
is NP-hard and solvable in subexponential time.
%We also show that the ETH is false if there exists a $\Delta$ such that
%$\VCSP(\Delta)$ is NP-hard and solvable in subexponential time.

\section{Preliminaries}

Let $\Gamma$ denote a finite set of finitary relations over $\B =
\{0,1\}$.  We call $\Gamma$ a {\em constraint language}. Given $R
\subseteq \B^{k}$ we let $\arity(R)= k$ denote its arity,
and similarly for functions. When $\Gamma = \{R\}$ we typically omit
the set notation and treat $R$ as a constraint language.

\subsection{Problem Definitions}
The {\em constraint satisfaction problem} over $\Gamma$ 
($\CSP(\Gamma)$) is defined as follows.

\smallskip

\noindent
{\sc Instance:} A set $V$ of variables and a set $C$ of constraint
applications $R(v_1,\ldots,v_k)$ where $R \in \Gamma$, $k=\arity(R)$, and
$v_1,\ldots,v_k \in V$.

\noindent
{\sc Question:} Is there a
% total
function $f:V \rightarrow \B$ such
that $(f(v_1),\ldots,f(v_k)) \in R$ for each 
$R(v_1,\ldots,v_k)$ in $C$?

\smallskip

For the Boolean domain this problem is typically denoted as
$\SAT(\Gamma)$. By $\SAT(\Gamma)$-$B$ we mean the $\SAT(\Gamma)$
problem restricted to instances where each variable can occur in at
most $B$ constraints. This restricted problem is occasionally useful
since each instance contains at most $B \, n$
constraints. 
% Note that this holds if and only if the function $f$ is surjective.
%
% and for larger domains
%this problem is known as {\em surjective CSP}.
%surjective. Note that
%for the Boolean domain this problem exactly correspond to the
%$\CSP^*(\Gamma)$ problem where we do not allow trivial solutions where
%all variables are assigned the same value~\cite{CrHe97}. 
The {\em weigthed maximum ones problem} over
$\Gamma$ ($\WMO(\Gamma)$) is an optimization version of
% the
$\SAT(\Gamma)$
% problem
 where we 
for an instance on variables $\{x_1,\dots,x_n\}$ and weights $w_i \in \mathbb{Q}_{\ge 0}$ want to find a solution $h$ for which $\sum_{i=1}^n w_i \, h(x_i)$ is maximal.
%maximize the sum 
%$\Sigma^{i = n}_{i = 1}w_i \cdot x_i$, where $w_i \in \mathbb{Q}_{\geq 0}$.
The {\em unweigthed maximum ones problem} ($\MO(\Gamma)$) is the
$\WMO(\Gamma)$ problem where all weights have the value 1.
%For
%optimization problems it is common to also consider weighted
%  problems} and we define {\em $\WMO(\Gamma)$} to be the $\MO(\Gamma)$
%problem where every instance $I$ over $n$ variables $x_1, \ldots, x_n$
%for every variable $x_i$ has an associated weight $w_i \in \mathbb{Q}$
%and the goal is to maximize the sum $\Sigma^{i = n}_{i = 1}w_i \cdot
%x_i$.  
A {\em finite-valued cost function} on $\B$ is a
% total
function
$f:\B^k \rightarrow \mathbb{Q}_{\geq 0}$.
The {\em valued constraint satisfaction problem} over
a finite set of finite-valued cost functions $\Delta$ ($\VCSP(\Delta)$) is defined as follows.
%: D^{\arity(f_i)}
%\mapsto \mathbb{Q}_{\geq 0}

\smallskip

\noindent
{\sc Instance:} A set $V = \{x_1, \ldots, x_n\}$ of variables and
the
objective function $f_I(x_1, \ldots, x_n) = \sum^{q}_{i=1}w_i \, f_i(\tup{x}^i)$ where, for every $1 \leq i \leq q, f_i \in \Delta,
\tup{x}^i \in V^{\arity(f_i)}$, and $w_i \in \mathbb{Q}_{\geq 0}$ is
a weight.

\noindent
%% {\sc Goal:} Find a function $h : V \mapsto D$ such that
%% $\sum^{q}_{i=1}w_i \, f_i(h(\tup{x}^i))$ is minimal (where $h$
%% is evaluated pointwise on $\tup{x}^i$).
{\sc Goal:} Find a function $h : V \to \B$ such that
$f_I(h(x_1),\dots,h(x_n))$ is minimal.

\smallskip

When the set of cost functions is singleton 
$\VCSP(\{f\})$ is written as $\VCSP(f)$. We let $\UWVCSP$ be the $\VCSP$
problem without weights and $\UWVCSP_d$ (for $d \geq 0$) denote the
$\UWVCSP$ problem restricted to instances containing at most $d \,
|\var(I)|$ constraints. Many optimization problems can be viewed as
$\VCSP(\Delta)$ problems for suitable $\Delta$: well-known examples
are the $\MAXCSP(\Gamma)$ and $\MINCSP(\Gamma)$ problems where the
number of satisfied constraints in a $\CSP$ instance are maximized or
minimized. For each $\Gamma$, there obviously exists sets of
cost functions $\Delta_{\min},\Delta_{\max}$
%(where the range of each function $f \in \Delta_{\min} \cup \Delta_{\max}$
%is $\{0,1\}$)
such that $\MINCSP(\Gamma)$ is polynomial-time equivalent to
$\VCSP(\Delta_{\min})$ and
$\MAXCSP(\Gamma)$ is polynomial-time equivalent to
$\VCSP(\Delta_{\max})$.
%One should note that $\UWVCSP$ allows a limited kind of weights since function applications
%may be repeated in the sum.
%This is in line with the common definition of $\MAXCSP(\Gamma)$ and $\MINCSP(\Gamma)$~\cite{Khanna:etal:sicomp2000}.
We have defined the problems
$\UWVCSP$, $\VCSP$, $\MO$ and $\WMO$ as optimization problems, but to
obtain a more uniform treatment we often view them as decision
problems, i.e.\ 
given $k$ we ask if there is a 
%measurement function less than (or greater than) this value.
solution with objective value $k$ or better.

%% Let $D$ be a finite domain and $\Gamma$ a set of functions
%% $f_i:D^{k_i} \rightarrow \{0,1\}$. Then, Max AW CSP$(\Gamma)$ is the
%% following maximisation problem.

%% \smallskip

%% \noindent
%% {\sc Instance.} A set of variables $V$ and a sum $\sum_{i=1}^m \rho_i \cdot f_i({\bf x}_i)$
%% where $\rho_i \in {\mathbb Q}$, $f_i \in \Gamma$, and ${\bf x}_i$ is a list of 
%% $k_i$ variables from $V$.

%% \smallskip

%% \noindent
%% {\sc Solution.} A function $\sigma : V \rightarrow D$.

%% \smallskip

%% \noindent
%% {\sc Measure.} $m(\sigma)=\sum_{i=1}^m \rho_i \cdot f_i(\sigma({\bf x}_i))$ where
%% $\sigma({\bf x}_i)$ is the list of elements from $D$ obtained by
%% applying $\sigma$ component-wise to ${\bf x}_i$.

%% \bigskip

%% \noindent
%% Let Max CSP$(\Gamma)$ denote the Max AW CSP$(\Gamma)$ problem under the restriction
%% that $\rho_i \in {\mathbb Q}_{\geq 0}$.
%% Define Min AW CSP$(\Gamma)$ and Min CSP$(\Gamma)$ analogously.

%% We use a logical notation when dealing with constraints over domain $D=\{0,1\}$.
%% For instance, we write $(p \vee \neg r)$ to denote the
%% the function $f(x,y)=1$ iff $p=1$ or $q=0$ (and $f(x,y)=0$ otherwise).
%% We also write $(p \vee \neg r):2$ to say that the sum shall contain the
%% function $(p \vee \neg r)$
%%  with
%% corresponding weight 2.
%% For $k \geq 3$, let $\Theta_k=\{\{(\sim x_1,\ldots, \sim x_l)\} \; | \; 1 \leq l \leq k\}$, i.e.
%% all functions on $\{0,1\}$ of arity $\leq k$ such that they evaluate to 1 on exactly
%% one variable assignment. (We use $\sim x$ to denote either $x$ or $\neg x$).
\subsection{Size-Preserving Reductions and Subexponential Time}
If $A$ is a computational problem we let $I(A)$ be the set of problem
instances and $\|I\|$ be the size of any $I \in I(A)$, i.e.\ 
the number of bits required to represent $I$.  Many problems can in
a natural way be viewed as problems of assigning values from a fixed
finite set to a collection of variables.  This is certainly the case
for $\SAT(\cdot)$, $\Maxones(\cdot)$ and $\VCSP(\cdot)$ but it
is also the case for various graph problems such as $\MAXCUT$ and
{\sc{Max Independent Set}}. We call problems of this kind {\em
  variable problems} and let $\var(I)$ denote the set of variables of
an instance $I$.
\begin{definition}
  Let $A_1$ and $A_2$ be variable problems in NP.
 The function $f$  from $I(A_1)$ to $I(A_2)$ is a {\em many-one
   linear variable reduction} (LV-reduction) with parameter $C \geq 0$
 if: (1) $I$ is a yes-instance of $A_1$ if and only if $f(I)$ is a
  yes-instance of $A_2$, (2) $|\var(f(I))| = C \cdot |\var(I)| + O(1)$, and 
  (3) $f(I)$ can be computed in time $O(\operatorname{poly}(\|I\|))$.
\end{definition}

LV-reductions can be seen as a restricted form of SERF-reductions~\cite{impagliazzo2001}. 
The term CV-reduction is
used to denote LV-reductions with parameter 1, and we write $A_1
\reduces A_2$ to denote that the problem $A_1$ has an CV-reduction to
$A_2$. If $A_1$ and $A_2$ are two NP-hard problems we say that $A_1$
is {\em at least as easy} as (or {\em not harder than}) $A_2$ if $A_1$
is solvable in $O(c^{|\var(I)|})$ time whenever $A_1$ is solvable in 
$O(c^{|\var(I)|})$ time. By definition if $A_1 \reduces A_2$ then $A_1$ is not
harder than $A_2$ but the converse is not true in general.
A problem solvable in time $O(2^{c \, |\var(I)|})$ for all $c > 0$
is a {\em subexponential problem}, and $\SE$ denotes 
%If every instance $I$ of an variable problem $A$ is solvable in time $O(2^{c \, |\var(I)|})$ for all $c > 0$
%then it is a {\em subexponential} problem. 
%We let $\SE$
%be 
the class of all variable problems solvable in subexponential time. It is
straightforward to prove that LV-reductions preserve subexponential
complexity in the sense that if $A$ is LV-reducible to $B$ then $A \in \SE$ if $B \in \SE$.
Naturally, $\SE$ can be defined using other complexity parameters than $|\var(I)|$~\cite{impagliazzo2001}.

%Let $\Gamma$ and $\Delta$ be
%two sets of functions. We say that Max-CSP($\Gamma$) is {\em as least
%  as easy} as Max-CSP($\Delta$) if Max-CSP($\Gamma$) is solvable in
%$O(c^{n})$ time whenever SAT($\Delta$) is solvable in $O(c^{n})$
%time. 
%If $\Gamma$ is a set of functions such that (1)
%MAX-Csp($\Gamma$) is NP-complete and (2) there does not exist any
%$\Delta$ such that MAX-Csp($\Gamma$) is harder than Max-CSP($\Delta$)
%then we say that Max-CSP($\Gamma$) is the {\em easiest NP-complete
%  Max-CSP($\cdot$) problem}.

\subsection{Clone Theory}
%In this section we introduce some notions from universal algebra
%and clone theory. 
An operation $f :\B^k \to \B$ is a {\em polymorphism} of a relation
$R$ if for every $\tup{t}^1,\dots,\tup{t}^k \in R$ it holds that
$f(\tup{t}^1,\dots,\tup{t}^k) \in R$, where $f$ is applied
element-wise. In this case $R$ is {\em
  closed}, or {\em invariant}, under $f$.
%If $f$ is an $n$-ary function and $R$ a relation
%with $m$ tuples it is possible to extend $f$ to operate over tuples
%from $R$ as follows:
%%\vspace{-10}
%%\begin{align*}
%\[f(t_1,\ldots ,t_n) = \big(f(t_{1,1}, \ldots, t_{n,1}),  \ldots 
%%& f(t_1[2], \dots, t_n[2]), \\ & \vdots \\
%f(t_{1,m}, \ldots, t_{n,m})\big),
%\]
%where $t_{i,j}$ denotes the $j$-th argument of the tuple $t_i \in
%R$. If $R$ is closed under $f$ we say that $f$ {\em preserves} $R$ or
%that $f$ is a {\em polymorphism} of $R$.
%
For a set of functions $\cc{F}{}{}$ we define Inv$(\cc{F}{}{})$ (often
abbreviated as $\cc{IF}{}{}$) to be the set of all relations invariant
under all functions in $\cc{F}{}{}$. Dually $\pol(\Gamma)$ for a set
of relations $\Gamma$ is defined to be the set of polymorphisms of
$\Gamma$.  Sets of the form $\pol(\Gamma)$ are known as {\em clones}
and sets of the form Inv$(\cc{F}{}{})$ are known as {\em
  co-clones}. The reader unfamiliar with these concepts is referred to
the textbook by Lau~\cite{lau2006}. The relationship between these
structures is made explicit in the following {\em Galois
  connection}~\cite{lau2006}.
\begin{theorem} \label{theorem:galois}
  Let $\Gamma$, $\Gamma'$ be sets of relations. Then
  $\mathrm{Inv}(\pol(\Gamma')) \subseteq \mathrm{Inv}(\pol(\Gamma))$ if and only if 
  $\pol(\Gamma) \subseteq \pol(\Gamma')$.
\end{theorem}
Co-clones can
equivalently be described as sets containing all relations $R$
definable through {\em primitive positive} (p.p.) implementations
over a constraint language $\Gamma$, i.e.\ definitions of the form
$R(x_1, \ldots, x_n) \equiv \exists y_1, \ldots, y_m\, . \,
R_1(\tup{x}^1) \wedge \ldots \wedge R_k(\tup{x}^k)$, where each $R_i
\in \Gamma \cup \{\eq\}$ and each $\tup{x}^i$ is a tuple over
$x_1,\ldots, x_n$, $y_1, \ldots, y_m$ and where $\eq = \{(0,0),(1,1)\}$. As a shorthand we let
$\cclone{\Gamma} = \mathrm{Inv}(\pol(\Gamma))$ for a constraint
language $\Gamma$, and as can be verified this is the smallest set of relations closed under
p.p.\ definitions over $\Gamma$. In this case $\Gamma$ is said to be a
{\em base} of $\cclone{\Gamma}$. 
It is known that if $\Gamma'$ is finite and $\pol(\Gamma) \subseteq
\pol(\Gamma')$ then $\CSP(\Gamma')$ is polynomial-time reducible to
$\CSP(\Gamma)$~\cite{jeavons1998}. With this fact and Post's
classification of all Boolean clones~\cite{pos41} Schaefer's dichotomy
theorem~\cite{Schaefer78} for $\SAT(\cdot)$ follows almost
immediately. See Figure~\ref{figure:clones} and
Table~\ref{table:clones} in Appendix~\ref{appendix:clones} for a
visualization of this lattice and a list of bases. The complexity of
$\Maxones(\Gamma)$ is also preserved under finite expansions with
relations p.p.\ definable in $\Gamma$, and hence follow the standard
Galois connection~\cite{Khanna:etal:sicomp2000}. Note however that
$\pol(\Gamma') \subseteq \pol(\Gamma)$ does not imply that
$\CSP(\Gamma')$ CV-reduces to $\CSP(\Gamma)$ or even that
$\CSP(\Gamma')$ LV-reduces to $\CSP(\Gamma)$ since the number of
constraints is not necessarily linearly bounded by the number of
variables.

To study these restricted classes of reductions we are therefore in need
of Galois connections with increased granularity. In Jonsson et
al.$\,$~\cite{Jonsson:etal:soda2013} the $\SAT(\cdot)$ problem is studied
with the Galois connection between closure under p.p.\ definitions
without existential quantification and {\em strong partial clones}.
We concentrate on the relational description and present the full
definitions of partial polymorphisms and the aforementioned Galois
connection in Appendix~\ref{appendix:weak_bases}.  If $R$ is an
$n$-ary Boolean relation and $\Gamma$ a constraint language then
$R$ has a {\em quantifier-free primitive positive} (q.p.p.)
implementation in $\Gamma$ if $R(x_1, \ldots, x_n) \equiv
R_1(\tup{x}^1) \wedge \ldots \wedge R_k(\tup{x}^k)$, where
each $R_i \in \Gamma \cup \{\eq\}$ and each $\tup{x}^i$ is a tuple
over $x_1,\ldots, x_n$. We use $\pcclone{\Gamma}$ to denote the
smallest set of relations closed under q.p.p.\ definability over $\Gamma$. If
$\cc{IC}{}{} = \pcclone{\cc{IC}{}{}}$ then $\cc{IC}{}{}$
is a {\em weak partial co-clone}. 
%We use the term weak partial
%co-clone to avoid confusion with partial co-clones used in other
%contexts (see Chapter 20.3 in Lau~\cite{lau2006}). 
In Jonsson et
al.$\,$~\cite{Jonsson:etal:soda2013} it is proven that if $\Gamma' \subseteq
\pcclone{\Gamma}$ and if $\Gamma$ and $\Gamma'$ are both finite
constraint languages then $\CSP(\Gamma') \reduces \CSP(\Gamma)$. It is
not hard to extend this result to the $\Maxones(\cdot)$ problem since it follows the standard Galois
connection, and therefore we use this fact without explicit proof.  A
{\em weak base} $R_w$ of a co-clone $\cc{IC}{}{}$ is then a base of
$\cc{IC}{}{}$ with the property that for any finite base $\Gamma$ of
$\cc{IC}{}{}$ it holds that $R_w \in \pcclone{\Gamma}$. In particular
this means that $\SAT(R_w)$ and $\Maxones(R_w)$ 
CV-reduce to $\SAT(\Gamma)$ and $\Maxones(\Gamma)$
for any base $\Gamma$ of $\cc{IC}{}{}$, and $R_w$
can therefore be seen as the easiest language in the co-clone. The
formal definition of a weak base is included in
Appendix~\ref{appendix:weak_bases} together with a table of weak bases
for all Boolean co-clones with a finite base. These weak bases have
the additional property that they can be implemented without the
equality relation~\cite{Lagerkvist2014}.

\subsection{Operations and Relations}
%(such as a conjunction of Boolean clauses) 
%or
%as a matrix representation with an $m \times n$-matrix containing the
%tuples of $R$ as rows in lexicographical order. Note that the
%ordering is only relevant to ensure that the representation is
%unambiguous. 
An operation $f$ is called
{\em arithmetical} if
$f(y,x,x)=f(y,x,y)=f(x,x,y)=y$ for every $x,y \in \B$. The 
$\max$ function is defined as $\max(x,y) = 0$ if $x = y = 0$ and 1
otherwise. We often express a Boolean relation $R$ as a logical
formula whose satisfying assignment corresponds to the tuples of
$R$. $\F$ and $\T$ are the two constant relations $\{(0)\}$ and
$\{(1)\}$ while $\neqq$ denotes inequality,
i.e.\ the relation $\{(0,1), (1,0)\}$. The
relation $\evenn{}{n}$ is defined as $\{(x_1, \ldots, x_n) \in
\B^n \mid \sum^{n}_{i = 1} x_i \text{ is even}\}$. The
relation $\oddn{}{n}$ is defined dually. The relations $\orn{}{n}$ and
$\nandn{}{n}$ are the relations corresponding to the clauses $(x_1
\vee \ldots \vee x_n)$ and $(\overbar{x_1} \vee \ldots \vee \overbar{x_n})$. For
any $n$-ary relation and $R$ we let $R_{m \neq}$, $1 \leq m \leq n$,  denote the $(n+m)$-ary
relation defined as $R_{m \neq}(x_1, \ldots, x_{n + m}) \equiv R(x_1,
\ldots, x_n) \wedge \neqq(x_1,x_{n+1}) \wedge \ldots \wedge \neqq(x_n,x_{n + m})$.
We use $R^{1/3}$ for the relation $\{(0,0,1),(0,1,0), (1,0,0)\}$.
Variables are typically named $x_1, \ldots, x_n$ or $x$ except when
they occur in positions where they are forced to take a particular
value, in which case they are named $c_0$ and
$c_1$ respectively to explicate that they are in essence constants. 
As convention $c_0$ and $c_1$ always occur in the last
positions in the arguments to a predicate. We now
see that $\Rel{II_2}(x_1, \ldots, x_6, c_0,c_1) \equiv
\Rddd(x_1,\ldots,x_6) \wedge \F(c_0) \wedge \T(c_1)$ and
$\Rel{IN_2}(x_1, \ldots, x_8) \equiv \evenn{4 \neq}{4}(x_1,\ldots,x_8)
\wedge (x_1x_4 \leftrightarrow x_2x_3)$ from
Table~\ref{table:weak_bases} in Appendix~\ref{appendix:clones} are the
two relations (where the tuples in the relations are listed as rows)
%{\small 
%\[  \Rel{II_2}(x_1, \ldots, x_6, c_0,c_1) = \begin{pmatrix} 
%      0 & 0 & 1 & 1 & 1 & 0 & 0 & 1 \\
%      0 & 1 & 0 & 1 & 0 & 0 & 0 & 1 \\
%      1 & 0 & 0 & 0 & 1 & 1 & 0 & 1
%  \end{pmatrix},\, \Rel{IN_2}(x_1, \ldots, x_8) =
%\begin{pmatrix} 
%      0 & 0 & 0 & 0 & 1 & 1 & 1 & 1 \\
%      0 & 0 & 1 & 1 & 1 & 1 & 0 & 0 \\
%      0 & 1 & 0 & 1 & 1 & 0 & 1 & 0 \\
%      1 & 1 & 1 & 1 & 0 & 0 & 0 & 0 \\
%      1 & 1 & 0 & 0 & 0 & 0 & 1 & 1 \\
%      1 & 0 & 1 & 0 & 0 & 1 & 0 & 1
%\end{pmatrix}.
%\]}
\begin{align*}
&\Rel{II_2} =
\left\{ \begin{smallmatrix} 
      0 & 0 & 1 & 1 & 1 & 0 & 0 & 1 \\
      0 & 1 & 0 & 1 & 0 & 1 & 0 & 1 \\
      1 & 0 & 0 & 0 & 1 & 1 & 0 & 1
\end{smallmatrix}\right\}
\quad\text{and}\quad
\Rel{IN_2} =
\left\{ \begin{smallmatrix} 
      0 & 0 & 0 & 0 & 1 & 1 & 1 & 1 \\
      0 & 0 & 1 & 1 & 1 & 1 & 0 & 0 \\
      0 & 1 & 0 & 1 & 1 & 0 & 1 & 0 \\
      1 & 1 & 1 & 1 & 0 & 0 & 0 & 0 \\
      1 & 1 & 0 & 0 & 0 & 0 & 1 & 1 \\
      1 & 0 & 1 & 0 & 0 & 1 & 0 & 1
\end{smallmatrix}\right\}.
\end{align*}
%
%where each row corresponds to a tuple in the relation.
%% and that $\Rel{IN_2}(x_1, \ldots, x_8) \equiv \evenn{4 \neq}{4}(x_1,\ldots,x_8) \wedge
%% x_1x_4 \leftrightarrow x_2x_3$ is the relation
%% {\small
%% \[  \begin{pmatrix} 
%%       0 & 0 & 0 & 0 & 1 & 1 & 1 & 1 \\
%%       0 & 0 & 1 & 1 & 1 & 1 & 0 & 0 \\
%%       0 & 1 & 0 & 1 & 1 & 0 & 1 & 0 \\
%%       1 & 1 & 1 & 1 & 0 & 0 & 0 & 0 \\
%%       1 & 1 & 0 & 0 & 0 & 0 & 1 & 1 \\
%%       1 & 0 & 1 & 0 & 0 & 1 & 0 & 1
%%   \end{pmatrix}. \]}

\section{The Easiest NP-Hard $\Maxones$ and
  $\VCSP$ Problems} \label{section:easy_problems}
We will now study the complexity of $\WMO$ and
$\VCSP$ with respect to CV-reductions. 
We remind the reader that constraint languages $\Gamma$ and sets
of cost functions $\Delta$ are always finite.
We prove that
for both these problems there is a single language which is
CV-reducible to every other NP-hard language. 
Out of the infinite number of candidate languages generating different
co-clones, the language $\{\Rel{II_2}\}$ defines the easiest
$\WMO(\cdot)$ problem, which is the same language as for
$\SAT(\cdot)$~\cite{Jonsson:etal:soda2013}. This might be contrary to
intuition since one could be led to believe that the co-clones in the
lower parts of the co-clone lattice, generated by very simple
languages where the corresponding $\SAT(\cdot)$ problem is in P, would
result in even easier problems.

\subsection {The $\Maxones$ Problem} \label{section:maxones}
Here we use a slight reformulation of Khanna et al.$\,$'s~\cite{Khanna:etal:sicomp2000}
complexity classification of the $\Maxones$ problem expressed in terms of
polymorphisms.

\begin{theorem}[\cite{Khanna:etal:sicomp2000}] \label{theorem:maxones}
Let $\Gamma$ be a finite Boolean constraint
language. $\Maxones(\Gamma)$ is in P if and only if $\Gamma$ is
$1$-closed, $\max$-closed, or closed under an arithmetical operation.
\end{theorem}

The theorem holds for both the weighted and the unweighted version of the
problem and showcases the strength of the algebraic
method since it not only eliminates all constraint languages
resulting in polynomial-time solvable problems, but also tells us
exactly which cases remain, and which properties they satisfy.

%which are 1-closed, $\max$-closed or closed under an arithmetical
%operation.

%implies that we only have to consider the co-clones
%where $\Maxones(\cdot)$ is NP-hard. 
%The proof of the following lemma is similar to that of
%Theorem~\ref{theorem:easiestsurcsp} and can be found in
%Appendix~\ref{appendix:umaxones}.

\begin{theorem} \label{theorem:umaxones}
$\MO(R)$ $\reduces \MO(\Gamma)$ for some $R \in \{R_{\cc{IS}{2}{1}}$,
  $R_{\cc{II}{}{2}}$, $R_{\cc{IN}{}{2}}$, $R_{\cc{IL}{}{0}}$, $R_{\cc{IL}{}{2}}$,
          $R_{\cc{IL}{}{3}}$, $R_{\cc{ID}{}{2}}\}$ whenever $\MO(\Gamma)$ is NP-hard.
\end{theorem}

\begin{proof}
By Theorem~\ref{theorem:maxones} in combination with
Table~\ref{table:clones} and Figure~\ref{figure:clones} in
Appendix~\ref{appendix:clones} it follows that $\MO(\Gamma)$ is NP-hard if
and only if $\cclone{\Gamma} \supseteq \cc{IS}{2}{1}$ or if
$\cclone{\Gamma} \in \{\cc{IL}{}{0}, \cc{IL}{}{3}, \cc{IL}{}{2},
\cc{IN}{}{2}\}$. In principle we then for every co-clone have to
decide which language is CV-reducible to every other base of the
co-clone, but since a weak base always have this property, we can
eliminate a lot of tedious work and directly consult the precomputed
relations in Table~\ref{table:weak_bases}.  From this we first see
that $\pcclone{R_{\cc{IS}{2}{1}}} \subset \pcclone{R_{\cc{IS}{n}{1}}}$,
$\pcclone{R_{\cc{IS}{2}{12}}} \subset \pcclone{R_{\cc{IS}{n}{12}}}$,
$\pcclone{R_{\cc{IS}{2}{11}}} \subset \pcclone{R_{\cc{IS}{n}{11}}}$
and $\pcclone{R_{\cc{IS}{2}{10}}} \subset
\pcclone{R_{\cc{IS}{n}{10}}}$ for every $n \geq 3$.  Hence in the four
infinite chains $\cc{IS}{n}{1}$, $\cc{IS}{n}{12}$, $\cc{IS}{n}{11}$,
$\cc{IS}{n}{10}$ we only have to consider the bottommost co-clones
$\cc{IS}{2}{1}$, $\cc{IS}{2}{12}$, $\cc{IS}{2}{11}$, $\cc{IS}{2}{10}$.
Observe that if $R$ and $R'$ satisfies $R(x_1,\dots,x_k) \Rightarrow
\exists y_0,y_1.R'(x_1,\dots,x_k,y_0,y_1) \land F(y_0) \land T(y_1)$
and $R'(x_1,\dots,x_k,y_0,y_1) \Rightarrow R(x_1,\dots,x_k) \land
F(y_0)$, and it moreover holds that $R'(x_1, \dots,x_k,y_0,y_1) \in
\pcclone{\Gamma}$, then $\MO(R) \reduces \MO(\Gamma)$, since we can
use $y_0$ and $y_1$ as global variables and because an optimal
solution to the instance we construct will always map $y_1$ to $1$ if
the original instance is satisfiable.  For $\Rel{IS^2_1}(x_1,x_2,c_0)$
we can q.p.p.\ define predicates $\Rel{IS^2_1}'(x_1,x_2,c_0,y_0,y_1)$
with $\Rel{IS^2_{12}}, \Rel{IS^2_{11}},
\Rel{IS^2_{10}}, R_{\cc{IE}{}{2}}, R_{\cc{IE}{}{0}}$ satisfying
these properties as follows:
\begin{itemize}
\item $\Rel{IS^2_1}'(x_1,x_2,c_0,y_0,y_1) \equiv \Rel{IS^2_{12}}(x_1,x_2,c_0,y_1) \land \Rel{IS^2_{12}}(x_1,x_2,y_0,y_1)$,
\item $\Rel{IS^2_1}'(x_1,x_2,c_0,y_0,y_1) \equiv \Rel{IS^2_{11}}(x_1, x_2, c_0,c_0) \land \Rel{IS^2_{11}}(x_1, x_2, y_0,y_0)$,
\item $\Rel{IS^2_1}'(x_1,x_2,c_0,y_0,y_1) \equiv \Rel{IS^2_{10}}(x_1,x_2,c_0,c_0,y_1) \land \Rel{IS^2_{10}}(x_1,x_2,c_0,y_0,y_1)$,
\item $\Rel{IS^2_1}'(x_1,x_2,c_0,y_0,y_1) \equiv R_{\cc{IE}{}{2}}(c_0,x_1,x_2,c_0,y_1) \land R_{\cc{IE}{}{2}}(c_0,x_1,x_2,y_0,y_1)$,
\item $\Rel{IS^2_1}'(x_1,x_2,c_0,y_0,y_1) \equiv R_{\cc{IE}{}{0}}(c_0, x_1,x_2, y_1, c_0) \land R_{\cc{IE}{}{0}}(y_0, x_1,x_2, y_1, y_0)$,
\end{itemize}
and similarly a relation $\Rel{II_2}'$ using $R_{\CoCloneII_0}$ as follows
$\Rel{II_2}'(x_1, x_2,x_3,x_4, x_5, x_6,c_0,c_1,y_0,y_1) \equiv
R_{\CoCloneII_0}(x_1,x_2,x_3,c_0) \land
R_{\CoCloneII_0}(c_0,c_1,y_1,y_0) \land R_{\CoCloneII_0}(x_1,x_4,y_1,
y_0) \land R_{\CoCloneII_0}(x_2,x_5,y_1, y_0) \land
R_{\CoCloneII_0}(x_3,x_6,\\y_1, y_0)$. By Figure~\ref{figure:clones} in
Appendix~\ref{appendix:clones} we
then see that the only remaining cases for $\Gamma$ when
$\cclone{\Gamma} \supset \cc{IS}{2}{1}$ is when $\cclone{\Gamma} =
\cc{II}{}{2}$ or when $\cclone{\Gamma} = \cc{ID}{}{2}$. This concludes
the proof.
\qed
\end{proof}

Using q.p.p.\ implementations to further decrease the set of relations
in Theorem~\ref{theorem:umaxones} appears difficult and
we therefore make use of more powerful implementations.
Let $\optsol(I)$ be the set of all optimal solutions
of a $\WMO(\Gamma)$ instance $I$.
%\begin{definition}[\cite{fourel,thapper}]
A relation $R$ has 
a \emph{weighted p.p.\ definition} (w.p.p.\ definition) ~\cite{fourel,thapper} in $\Gamma$ if there
exists an instance $I$ of $\WMO(\Gamma)$ on variables $V$ such that
$R=\{ (\phi(v_1),\dots,\phi(v_m)) \mid \phi \in \optsol(I) \}$ for some
$v_1,\dots,v_m \in V$.
The set of all relations w.p.p.\ definable in $\Gamma$ is denoted
$\langle\Gamma\rangle_w$ and we furthermore have that if $\Gamma'
\subseteq \langle\Gamma\rangle_{w}$ is a finite then $\WMO(\Gamma')$
is polynomial-time reducible to $\WMO(\Gamma)$~\cite{fourel,thapper}.
%% \begin{theorem}[\cite{fourel,thapper}] \label{wpp}
%% Let
%% $\Gamma' \subseteq \langle\Gamma\rangle_{w}$ be a finite constraint language. Then, the problem
%% $\WMO(\Gamma')$ is polynomial-time reducible to $\WMO(\Gamma)$.
%% \end{theorem}
If there is a $\WMO(\Gamma)$ instance $I$ on $V$ such that $R=\{
(\phi(v_1),\dots,\phi(v_m)) \mid \phi \in \optsol(I) \}$ for
$v_1,\dots,v_m \in V$ satisfying $\{v_1,\dots,v_m\}=V$, then we say
that $R$ is q.w.p.p.\ definable in $\Gamma$.  We use
$\langle\Gamma\rangle_{\nexists,w}$ for set of all relations
q.w.p.p.\ definable in $\Gamma$. It is not hard to check that if
$\Gamma' \subseteq \langle\Gamma\rangle_{\nexists,w}$, then every
instance is mapped to an instance of equally many variables --- hence 
$\WMO(\Gamma')$ is CV-reducible to $\WMO(\Gamma)$ whenever $\Gamma'$ is
finite.
%This gives the following corollary.

%\begin{corollary}
%\label{cor:qwpp:red}
%Let $\Gamma' \subseteq \langle\Gamma\rangle_{\nexists,w}$ be a finite
%set
%constraint language.
%Then, $\WMO(\Gamma')$
%$\reduces$
%is CV-reducible to
%$\WMO(\Gamma)$.
%\end{corollary}

\begin{theorem} \label{theorem:easiestwmo}
  Let $\Gamma$ be a constraint language such that $\WMO(\Gamma)$ is
  NP-hard. Then it holds that $\WMO(\Rel{II_2}) \reduces \WMO(\Gamma)$.
\end{theorem}

\begin{proof}
We utilize q.w.p.p.\ definitions and note that the following holds.
\begin{align*}
\Rel{II_2} &= \smash{\textstyle{\argmax_{\tup{x} \in \B^8 :
  (\tup{x}_7,\tup{x}_1,\tup{x}_2,\tup{x}_6,\tup{x}_8,\tup{x}_4,\tup{x}_5,\tup{x}_3)\in \Rel{IN_2}}  \tup{x}_8}},\\
\Rel{II_2} &= \smash{\textstyle{\argmax_{\tup{x} \in \B^8 :
  (\tup{x}_5,\tup{x}_4,\tup{x}_2,\tup{x}_1,\tup{x}_7,\tup{x}_8),
  (\tup{x}_6,\tup{x}_4,\tup{x}_3,\tup{x}_1,\tup{x}_7,\tup{x}_8),
  (\tup{x}_6,\tup{x}_5,\tup{x}_3,\tup{x}_4,\tup{x}_7,\tup{x}_8) \in
  \Rel{ID_2}} (\tup{x}_1+\tup{x}_2+\tup{x}_3)}},\\
\Rel{II_2} &= \smash{\textstyle{\argmax_{\tup{x} \in \B^8 :
  (\tup{x}_4,\tup{x}_5,\tup{x}_6,\tup{x}_1,\tup{x}_2,\tup{x}_3,\tup{x}_7,\tup{x}_8)\in
  \Rel{IL_2}} (\tup{x}_4+\tup{x}_5+\tup{x}_6)}},\\
\Rel{IL_2} &= \smash{\textstyle{\argmax_{\tup{x} \in \B^8:
  (\tup{x}_7,\tup{x}_1,\tup{x}_2,\tup{x}_3,\tup{x}_8,\tup{x}_4,\tup{x}_5,\tup{x}_6)\in
  \Rel{IL_3}} \tup{x}_8}}, \\
\Rel{IL_2} &= \smash{\textstyle{\argmax_{\tup{x} \in \B^8 :
  (\tup{x}_4,\tup{x}_5,\tup{x}_6,\tup{x}_7),
  (\tup{x}_8,\tup{x}_1,\tup{x}_4,\tup{x}_7),
  (\tup{x}_8,\tup{x}_2,\tup{x}_5,\tup{x}_7),
  (\tup{x}_8,\tup{x}_3,\tup{x}_6,\tup{x}_7) \in \Rel{IL_0}}
    \tup{x}_8}}, \\
\Rel{II_2} &= \smash{\textstyle{\argmax_{\tup{x} \in \B^8 :
  (\tup{x}_1,\tup{x}_2,\tup{x}_7), (\tup{x}_1,\tup{x}_3,\tup{x}_7),
  (\tup{x}_2,\tup{x}_3,\tup{x}_7), (\tup{x}_1,\tup{x}_4,\tup{x}_7),
  (\tup{x}_2,\tup{x}_5,\tup{x}_7), (\tup{x}_3,\tup{x}_6,\tup{x}_7) \in
  \Rel{IS_1^2}} (\tup{x}_1+\dots+\tup{x}_8)}}.
\end{align*}
Hence, $\Rel{II_2} \in \langle R\rangle_{\nexists,w}$ for every $R \in
\{R_{\cc{IS}{2}{1}}, R_{\cc{IN}{}{2}}, R_{\cc{IL}{}{0}},
R_{\cc{IL}{}{2}}, R_{\cc{IL}{}{3}}, R_{\cc{ID}{}{2}}\}$ which by
Theorem~\ref{theorem:umaxones} completes the proof.
\qed
\end{proof}

\subsection {The $\VCSP$ Problem} \label{section:vcsp}
Since $\VCSP$ does not adhere to the standard Galois connection in Theorem~\ref{theorem:galois},
the weak base method is not applicable and alternative methods are required.
For this purpose we use {\em multimorphisms} from Cohen et al.~\cite{vcsp22}.
%We call a function $f : D^k \to \mathbb{Q}_{\ge 0} \cup \{\infty\}$ a \emph{cost function}.
%
%\begin{definition}[\cite{vcsp22}]
Let $\Delta$ be a set of cost functions on $\B$, let $p$ be a unary
operation on $\B$, and let $f,g$ be binary operations on $\B$.
We say that $\Delta$ admits the binary {\em multimorphism} $(f,g)$ if it holds that
$\nu(f(x,y))+\nu(g(x,y)) \le \nu(x)+\nu(y)$ for every $\nu \in \Delta$ and $x,y \in \B^{\arity(\nu)}$.
Similarly $\Delta$ admits the unary {\em multimorphism} $(p)$ if it holds that
$\nu(p(x)) \le \nu(x)$ for every $\nu \in \Delta$ and $x \in \B^{\arity(\nu)}$.
%
%\begin{lemma}[\cite{vcsp22}]
%Let $\Delta$ be a set of finite-valued cost functions on $\B$. Then
%if $\Delta$ admits the unary $(0)$-multimorphism, the unary
%$(1)$-multimorphism or the binary $(\min,\max)$-multimorphism, then
%$\VCSP(\Delta)$ is in PO. Otherwise $\VCSP(\Delta)$ is NP-hard~\cite{vcsp22}.
%\end{lemma}
%
%We remark that {\em fractional polymorphisms} have been successfully
%applied to the $\VCSP(\cdot)$ problem for larger domains and that the
%complexity of $\VCSP(\Gamma)$ is completely determined for all
%finite-valued set of cost functions $\Gamma$ over a finite
%domain~\cite{thapper2013}.
%
%Let $\Gamma$ be a finite set of finitary functions on $\{0,1\}$ taking
%rational values.
%
%\begin{theorem} 
%Let $F$ be a set of functions such that VCSP() VCSP($\{xor\}$) is the easiest NP-complete VCSP($\cdot$) problem.
%\end{theorem}
%
%
%
Recall that the function $f_{\ne}$ equals $\{(0,0)\mapsto 1, (0,1)
\mapsto 0, (1,0) \mapsto 0, (1,1) \mapsto 1\}$ and that the
minimisation problem $\VCSP(f_{\ne})$ and the maximisation problem
{\sc Max Cut} are trivially CV-reducible to each other.  We will make
use of (a variant of) the concept of
\emph{expressibility}~\cite{vcsp22}.  We say that a cost function $g$
is \emph{$\nexists$-expressible} in $\Delta$ if $g(x_1,\dots,x_n) =
\sum_{i}w_i f_i(\tup{s}^i) + w$ for some tuples $\tup{s}^i$ over
$\{x_1,\dots,x_n\}$, weights $w_i \in \mathbb{Q}_{\ge 0}$, $w \in
\mathbb{Q}$ and $f_i \in \Delta$.  It is not hard to see that if every
function in a finite set $\Delta'$ is $\nexists$-expressible in
$\Delta$, then $\VCSP(\Delta') \reduces \VCSP(\Delta)$. Note that if the constants 0 and 1 are
expressible in $\Delta$ then we may allow tuples $\tup{s}^i$ over
$\{x_1,\dots,x_n,0,1\}$, and still obtain a CV-reduction.
%Note that since we use CV-reductions we additionaly
%allow the usage of a constant number of global variables.  

\begin{theorem} \label{theorem:easiestvcsp}
Let $\Delta$ be a set of finite-valued cost functions on $\B$.
If the problem $\VCSP(\Delta)$ is NP-hard, then $\VCSP(f_{\ne}) \reduces \VCSP(\Delta)$.
\end{theorem}
 \begin{proof}
Since $\VCSP(\Delta)$ is NP-hard (and since we assume P $\ne$ NP) we
know that $\Delta$ does not admit the unary $(0)$-multimorphism or the
unary $(1)$-multimorphism~\cite{vcsp22}. Therefore there are $g, h \in
\Delta$ and $\tup{u} \in \B^{\arity(g)}$, $\tup{v} \in \B^{\arity(h)}$
such that $g(\tup{0})>g(\tup{u})$ and $h(\tup{1})>h(\tup{v})$.  Let
$\tup{w} \in \argmin_{\tup{x} \in \B^b}
(g(\tup{x}_1,\dots,\tup{x}_a)+h(\tup{x}_{a+1},\dots,\tup{x}_b))$ and
then define $o(x,y) = g(z_1,\dots,z_a)+h(z_{a+1},\dots,z_b)$ where
$z_i=x$ if $\tup{w}_i=0$ and $z_i=y$ otherwise.  Clearly $(0,1) \in
\argmin_{\tup{x} \in \B^2} o(\tup{x})$, $o(0,1) < o(0,0)$, and $o(0,1)
< o(1,1)$.  We will show that we always can force two fresh variables
$v_0$ and $v_1$ to $0$ and $1$, respectively. If $o(0,0) \ne o(1,1)$,
then assume without loss of generality that $o(0,0)<o(1,1)$.  In this
case we force $v_0$ to $0$ with the (sufficiently weighted) term
$o(v_0,v_0)$. Define $g'(x) = g(z_1,\dots,z_{\arity(g)})$ where
$z_i=x$ if $u_i=1$ and $z_i=v_0$ otherwise. Note that $g'(1) < g'(0)$
which means that we can force $v_1$ to $1$. Otherwise $o(0,0) =
o(1,1)$. If $o(0,1) = o(1,0)$, then $f_{\ne} = \alpha_1 o +
\alpha_2$, otherwise assume without loss of generality that
$o(0,1)<o(1,0)$. In this case $v_0,v_1$ can be forced to $0,1$ with the help
of the (sufficiently weighted) term $o(v_0,v_1)$.

We also know that $\Delta$ does not
admit the $(\min,\max)$-multimorphism~\cite{vcsp22} since
$\VCSP(\Delta)$ is NP-hard by assumption.  Hence, there exists a $k$-ary
function $f \in \Delta$ and $\tup{s},\tup{t} \in \B^k$ such that
$f(\min(\tup{s},\tup{t}))+f(\max(\tup{s},\tup{t})) > f(\tup{s})+f(\tup{t})$.
Let $f_1(x) = \alpha_1 o(v_0,x) + \alpha_2$ for some $\alpha_1 \in
\mathbb{Q}_{\ge 0}$ and $\alpha_2 \in \mathbb{Q}$ such that
$f_1(1)=0$ and $f_1(0)=1$.  Let also $g(x,y) = f(z_1,\dots,z_k)$ where
$z_i=v_1$ if $\min(\tup{s}_i,\tup{t}_i)=1$, $z_i=v_0$ if
$\max(\tup{s}_i,\tup{t}_i)=0$, $z_i=x$ if $s_i > t_i$ and $z_i=y$
otherwise.  Note that $g(0,0)=f(\min(\tup{s},\tup{t}))$,
$g(1,1)=f(\max(\tup{s},\tup{t}))$, $g(1,0)=f(\tup{s})$ and
$g(0,1)=f(\tup{t})$.  Set $h(x,y) = g(x,y)+g(y,x)$.  Now
$h(0,1)=h(1,0) < \frac{1}{2} ( h(0,0)+h(1,1) )$.  If $h(0,0)=h(1,1)$,
then $f_{\ne} = \alpha_1 h + \alpha_2$ for some $\alpha_1 \in
\mathbb{Q}_{\ge 0}$ and $\alpha_2 \in \mathbb{Q}$.  Hence, we can without
loss of generality assume that $h(1,1) - h(0,0) = 2$.  Note now that
$h'(x,y) = f_1(x) + f_1(y) + h(x,y)$ satisfies $h'(0,0) = h'(1,1) =
\frac{1}{2}( h(0,0) + h(1,1) + 2 ) $ and $h'(0,1) = h'(1,0) =
\frac{1}{2} (2 + h(0,1) + h(1,0))$.  Hence, $h'(0,0) = h'(1,1) >
h'(0,1) = h'(1,0)$.  So $f_{\ne} = \alpha_1 h' + \alpha_2$ for some
$\alpha_1 \in \mathbb{Q}_{\ge 0}$ and $\alpha_2 \in \mathbb{Q}$.  \qed
\end{proof}

%Let $f_{\ne} = \{ (0,0) \mapsto 1, (0,1) \mapsto 0, (1,0) \mapsto 0,
%(1,1) \mapsto 1 \}$, $c_0 = \{0\}$ and $c_1 = \{1\}$. 

\subsection {The Broader Picture}
\label{sec:broadpicture}

Theorems \ref{theorem:easiestwmo} and~\ref{theorem:easiestvcsp} does
not describe the relative complexity between the $\SAT(\cdot)$, {\sc
  Max-Ones}$(\cdot)$ and $\VCSP(\cdot)$ problems. However we readily
see (1) that $\SAT(\Rel{II_2}) \reduces \WMO(\Rel{II_2})$, and (2)
that $\WMO(\Rel{II_2}) \reduces \Wprefix${\sc{-Max Independent Set}}
since $\Wprefix${\sc{-Max Independent Set}} can be expressed by
$\WMO(\Nand^2)$. The problem {\sc W-Max-Ones}$(\Nand^2)$ is in turn
expressible by $\MAXCSP(\{\Nand^2,$ $\T,\F\})$.  To show that
$\Wprefix${\sc{-Max Independent Set}} $\reduces \VCSP(f_{\ne})$ it is
in fact, since $\MAXCSP(\neqq)$ and $\VCSP(f_{\ne})$ is the same
problem, sufficient to show that {\sc Max-CSP}$(\{\Nand^2,\T,\F\})
\reduces \MAXCSP(\neqq)$.  We do this as follows.  Let $v_0$ and $v_1$
be two global variables.  We force $v_0$ and $v_1$ to be mapped to
different values by assigning a sufficiently high weight to the
constraint $\neqq(v_0,v_1)$.  It then follows that $\T(x) =
\neqq(x,v_0)$, $\F(x) = \neqq(x,v_1)$ and $\Nand^2(x,y) = \frac{1}{2}
( \neqq(x,y) + \F(x) + \F(y) )$ and we are done. It follows from this
proof that $\MAXCSP(\{\Nand^2,\T,\F\})$ and $\VCSP(f_{\ne})$ are
mutually CV-interreducible.  Since $\MAXCSP(\{\Nand^2,\T,\F\})$ can
also be formulated as a $\VCSP$ it follows that $\VCSP(\cdot)$ does
not have a unique easiest set of cost functions. The complexity
results are summarized in Figure~\ref{figure:compres}. Some trivial
inclusions are omitted in the figure: for example it holds that $\SAT(\Gamma)
\reduces \WMO(\Gamma)$ for all $\Gamma$.

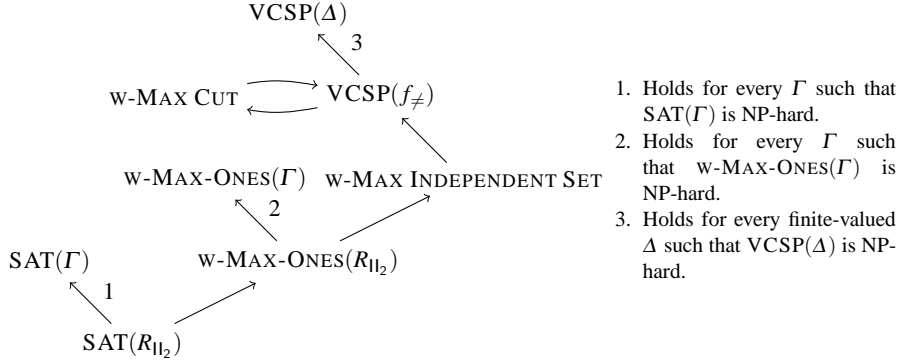
\begin{figure}
%\centering
\begin{minipage}[c]{0.65\textwidth}
%  \centering
    \begin{tikzpicture}[->,scale=1.1]

%        \node (surcspa) at (4.5,0) {$\CSPSTAR(\{\Rel{II_2}\})$};
%        \node (surcspb) at (6.5,1) {$\CSPSTAR(\Gamma)$};

%        \node (surcspa) at (-0.5,0) {$\CSPSTAR(\Rel{II_2})$};
%        \node (surcspb) at (-1.5,1) {$\CSPSTAR(\Gamma)$};

        \node (cspa) at (2,0) {$\SAT(\Rel{II_2})$};
        \node (cspb) at (1,1) {$\SAT(\Gamma)$};

        \node (maxonesa) at (4,1) {$\WMO(\Rel{II_2})$};
        \node (maxonesb) at (3,2) {$\WMO(\Gamma)$};

        \node (misa) at (6,2) {$\Wprefix$-{\sc{Max Independent Set}}};
        \node (maxcut) at (5,3) {$\VCSP(f_{\ne})$};
        \node (vcsp) at (4,4) {$\VCSP(\Delta)$};

%        \node (maxcutb) at (7.5,3) {(W.) Max-Cut};
        \node (maxcutb) at (2.5,3) {$\Wprefix$-{\sc{Max Cut}}};

        \path (cspa) edge node [above right, pos=0.4] {1} (cspb);
%        \path (surcspa) edge node [below right, pos=0.5] {2} (surcspb);
%        \path (surcspa) edge node [above right, pos=0.4] {2} (surcspb);
        \draw (cspa) to (maxonesa);

        \path (maxonesa) edge node [above right, pos=0.4] {2} (maxonesb);

        \draw (maxonesa) to (misa);
        \draw (misa) to (maxcut);

%        \draw (maxcut) to[out=10,in=170] (maxcutb);
%        \draw (maxcutb) to[out=190,in=-10] (maxcut);

        \draw (maxcutb) to[out=10,in=170] (maxcut);
        \draw (maxcut) to[out=190,in=-10] (maxcutb);

        %\draw (cspa) to[out=10,in=170] (surcspa);
        %\draw (surcspa) to[out=190,in=-10] (cspa);

%        \draw (surcspa) to[out=10,in=170] (cspa);
%        \draw (cspa) to[out=190,in=-10] (surcspa);

        \path (maxcut) edge node [above right, pos=0.4] {3} (vcsp);
    \end{tikzpicture}
\end{minipage}%
\begin{minipage}[c]{0.32\textwidth}
%\begin{description}
\scalebox{0.9}{
\begin{minipage}[c]{1.1111\textwidth}
\begin{enumerate}
\item%[1]
 Holds for every $\Gamma$ such that $\SAT(\Gamma)$ is NP-hard.
%\item%[2]
%Holds for every $\Gamma$ such that $\CSPSTAR(\Gamma)$ is NP-hard.
\item%[3]
Holds for every $\Gamma$ such that $\WMO(\Gamma)$ is NP-hard.
\item%[4]
Holds for every finite-valued $\Delta$ such that $\VCSP(\Delta)$ is NP-hard.
\end{enumerate}
%\end{description}
\end{minipage}
}
\end{minipage}
\caption{The complexity landscape of some Boolean optimization and satisfiability
  problems. A directed arrow from one node $A$ to $B$ means that $A \reduces B$.}
\label{figure:compres}
\end{figure}

\section{Subexponential Time and the Exponential-Time Hypothesis} \label{section:eth}
%The exponential-time hypothesis states that $k$-$\SAT$ is not
%subexponential when $k \geq 3$~\cite{impagliazzo2001b}. 
The exponential-time hypothesis states that $3$-$\SAT \notin \SE$
\cite{impagliazzo2001b}. We remind the reader that the ETH can be
based on different size parameters (such as the number of variables or
the number of clauses) and that these different definitions often
coincide~\cite{impagliazzo2001}.  In this section we investigate the
consequences of the ETH for the $\MO$ and $\UWVCSP$ problems. A direct
consequence of Section~\ref{section:easy_problems} is that if there
exists any finite constraint language $\Gamma$ or set of cost
functions $\Delta$ such that $\WMO(\Gamma)$ or $\VCSP(\Delta)$ is
NP-hard and in $\SE$, then $\SAT(\Rel{II_2})$ is in $\SE$ which
implies that the ETH is false~\cite{Jonsson:etal:soda2013}.
%any of these problems
%are NP-hard and in {\sc Subexp} then $\CSP(\Rel{II_2}) \in
%\textsc{Subexp}$. 
The other direction is interesting too since it highlights the
likelihood of subexponential time algorithms for the problems,
relative to the ETH.

\begin{lemma} \label{lemma:ethmo}
  If $\MO(\Gamma)$ is in $\SE$ for some finite constraint 
 languages $\Gamma$ such that $\MO(\Gamma)$ is NP-hard, then the
 ETH is false.
\end{lemma}

\begin{proof}
  From Jonsson et al.$\,$~\cite{Jonsson:etal:soda2013} it follows that
  $3$-$\SAT$ is in $\SE$ if and only if $\SAT(\Rel{II_2})$-$2$ is
  in $\SE$. Combining this with 
  Theorem~\ref{theorem:umaxones} we only have to prove that
  $\SAT(\Rel{II_2})$-$2$ LV-reduces to $\MO(R)$ for $R \in
  \{R_{\cc{IS}{2}{1}}, R_{\cc{IN}{}{2}}, R_{\cc{IL}{}{0}},
  R_{\cc{IL}{}{2}}, R_{\cc{IL}{}{3}}, R_{\cc{ID}{}{2}}\}$.  We provide
  an illustrative reduction from $\SAT(\Rel{II_2})$-$2$ to
  $\MO(\Rel{IS^2_1})$; the remaining reductions are 
  presented in Lemmas~\ref{mored1}--\ref{mored5} in Appendix~\ref{appendix:eth}. 
  Since $\Rel{IS^2_1}$ is the $\Nand$ relation with
  one additional constant column, the $\MO(\Rel{IS^2_1})$ problem is
  basically the maximum independent set problem or, equivalently,
  the maximum clique problem in the complement graph. Given an instance $I$ of $\CSP(\Rel{II_2})$-$2$ we
  create for every constraint $3$ vertices, one corresponding to each
  feasible assignment of values to the variables occurring in the
  constraint.  We add edges between all pairs of vertices that
  are not inconsistent and that do not correspond to the same
  constraint.
  The instance $I$ is satisfied if and only if there is a clique of
  size $m$ where $m$ is the number of constraints in $I$.  Since $m
  \leq 2n$ this implies that the number of vertices is $\le
  2n$. 
\qed
\end{proof}

\begin{theorem} \label{theorem:eth}
  The following statements are equivalent.
\begin{enumerate}
\item \label{i1}
  The exponential-time hypothesis is false.
\item \label{i4}
  $\MO(\Gamma) \in \SE$ for every finite $\Gamma$.
\item \label{i5}
  $\MO(\Gamma) \in \SE$ for some finite $\Gamma$ such that $\MO(\Gamma)$ is NP-hard.
%\item \label{i6}
%  $\WMO(\Gamma) \in {\textsc{Subexp}}$ for some finite constraint
%  language $\Gamma$ such that $\WMO(\Gamma)$ is NP-hard.
%\item \label{i4}
%  $\MinO(\Gamma) \in {\textsc{Subexp}}$ for every finite constraint language
%  $\Gamma$.
\item \label{i6}
  $\UWVCSP(\Delta)_d \in \SE$ for every finite set of finite-valued
  cost functions $\Delta$ and $d \geq 0$.
%\item \label{i8}
%  $\VCSP(\Gamma) \in {\textsc{Subexp}}$ for some finite set of
%  finite-valued cost functions $\Gamma$ such that $\UWVCSP(\Gamma)$
%  is NP-hard.
\end{enumerate}
\end{theorem}
\begin{proof}
The implication $\ref{i1} \Rightarrow \ref{i4}$ follows from
Lemma~\ref{lemma:ethmaxones} in Appendix~\ref{appendix:eth}, $\ref{i4}
\Rightarrow \ref{i5}$ is trivial, and $\ref{i5} \Rightarrow \ref{i1}$
follows by Lemma~\ref{lemma:ethmo}.  The implication $\ref{i4}
\Rightarrow \ref{i6}$ follows from Lemma~\ref{lemma:ethvcsp} in
Appendix~\ref{appendix:eth}.  We finish the proof by showing $\ref{i6}
\Rightarrow \ref{i1}$.  Let $I=(V,C)$ be an instance of
$\SAT(\Rel{II_2})$-$2$.  Note that $I$ contains at most $2\,|V|$
constraints.  Let $f$ be the function defined by $f(\tup{x})=0$ if
$\tup{x} \in \Rel{II_2}$ and $f(\tup{x})=1$ otherwise.  Create an
instance of $\UWVCSP_2(f)$ by, for every constraint $C_i =
\Rel{II_2}(x_1,\dots,x_8) \in C$, adding to the cost function the term
$f(x_1,\dots,x_8)$.  This instance has a solution with objective value
$0$ if and only if $I$ is satisfiable.  Hence, $\SAT(\Rel{II_2})$-$2
\in \SE$ which contradicts the ETH~\cite{Jonsson:etal:soda2013}. \qed
\end{proof}

\section{Future Research}
%Here we touch upon some directions for future research.
%\smallskip
%\noindent
  {\bf Other problems.} 
  The weak base method naturally lends itself to other problems
  parameterized by constraint languages. In general, one has to 
  consider all co-clones where the problem is NP-hard, take the
  weak bases for these co-clones and find out which of these are
  CV-reducible to the other cases. 
The last step is typically the most challenging --- this was
demonstrated by the $\MO$ problems where we had to introduce
q.w.p.p. implementations.
An example of an interesting problem where this strategy works is
the {\em
    non-trivial} $\SAT$ problem ($\SAT^*(\Gamma)$), i.e.\ the problem
  of deciding whether a given instance has a solution in which not all
  variables are mapped to the same value. This problem is NP-hard in
  exactly six cases~\cite{CrHe97} and by following the aforementioned
  procedure one can prove that the relation $\Rel{II_2}$ results in
  the easiest NP-hard $\SAT^*(\Gamma)$ problem. Since
  $\SAT^*(\Rel{II_2})$ is in fact the same problem as
  $\SAT(\Rel{II_2})$ this shows that restricting solutions to
  non-trivial solutions does not make the satisfiability problem
  easier. This result can also be extended to the co-NP-hard {\em
    implication problem}~\cite{CrHe97} and we believe that similar methods can
  also be applied to give new insights into the complexity of
  e.g. {\em enumeration}, which also follows the same complexity classification~\cite{CrHe97}.
 Such results would naturally give us insights into the structure of NP but also into the applicability of clone-based methods.

\smallskip

\noindent
{\bf Weighted versus unweighted problems.} 
Theorem~\ref{theorem:eth} only applies to unweighted
problems and
lifting these results to the weighted case does
not appear straightforward.
%There are two major obstacles in lifting this result to the
%unrestricted problems $\VCSP(\cdot)$ and $\WMO(\cdot)$: first,
%instances may contain a large number of constraints compared to the
%%number of variables; second, weights make it difficult to construct
%reductions which only increases the number of variables by a linear
%factor. 
We believe that some of these obstacles could be overcome with
generalized sparsification techniques. 
%and that a good line of
%attack is to prove that the $\problemFont{Max-Cut}$ problem (which is
%just a reformulation of $\VCSP(f_{\ne})$) is in $\SE$ if
%the ETH is false. 
We provide an example by proving that
if any NP-hard $\WMO(\Gamma)$ problem is
in $\SE$, then $\MAXCUT$ can be approximated 
within a multiplicative error of $(1 \pm \epsilon)$ (for any $\epsilon > 0$) in subexponential time.
Assume that $\WMO(\Gamma)$ is NP-hard and a member of $\SE$, and arbitrarily choose $\epsilon > 0$.
Let $\MAXCUT_c$ be the $\MAXCUT$ problem restricted to graphs $G =
(V,E)$ where $|E| \leq c \cdot |V|$. 
We first prove that $\MAXCUT_c$ is in $\SE$ for arbitrary $c \geq 0$.
By Theorem~\ref{theorem:easiestwmo}, we infer that
$\WMO(\Rel{II_2})$ is in $\SE$. Given an instance $(V,E)$
of $\MAXCUT_c$, one can introduce one fresh variable $x_v$ for each
$v \in V$ and one fresh variable $x_e$ for each edge $e \in E$. For each
edge $e = (v,w)$, we then constrain the
variables $x_v,x_w$ and $x_e$ as $R(x_v,x_w,x_e)$ where $R =
\{(0,0,0),(0,1,1),(1,0,1),(1,1,0)\} \in \cclone{\Rel{II_2}}$. It can
then be verified that, for an optimal solution $h$, that the maximum value of $\sum_{e \in E} w_e h(x_e)$
(where $w_e$ is the weight associated with the edge $e$) equals the
weight of a maximum cut in $(V,E)$. This is an LV-reduction 
since $|E| = c \cdot |V|$.
Now consider an instance
$(V,E)$ of the unrestricted $\MAXCUT$ problem. By Batson et al.~\cite{batson2012},
 we can (in polynomial time) compute a
{\em cut sparsifier} $(V',E')$ with only $D_{\epsilon} \cdot
n/\epsilon^2$ edges (where $D_{\epsilon}$ is a constant depending only on
$\epsilon$), which approximately preserves the value of the
maximum cut of $(V,E)$ to within a multiplicative error of $(1 \pm
\epsilon)$. By using the LV-reduction above from $\MAXCUT_{D_{\epsilon}/\epsilon^2}$
to $\WMO(\Gamma)$, it follows
that we can approximate the maximum cut of $(V,E)$ within $(1 \pm \epsilon)$
in subexponential time.

%Batson et al~\cite{batson2012} have shown that with only
%$D_{\epsilon} \cdot n/\epsilon^2$ (where $D_{\epsilon}$ is a constant
%depending only on $\epsilon$) edges can be constructed in $O(|V|^3|E|)$ time
%for fixed $\epsilon > 0$ which approximately preserves the value of
%every maximum cut to within a multiplicative error of $(1 \pm
%\epsilon)$

\bibliography{references}
\bibliographystyle{plain}

\newpage
\appendix
\section{Appendix}
%\subsection{The lattice of Boolean clones}

%\newpage

\subsection{Bases of Boolean Clones and the Clone Lattice}\label{appendix:clones}
In Table~\ref{table:clones} we present a full table of bases for all
Boolean clones. These were first introduced by Post~\cite{pos41} and
the lattice is hence known as {\em Post's lattice}.
%These were first introduced by Post~\cite{pos41} but
%the modern terminology stems from Reith and
%Wagner~\cite{reith2000}. 
It is visualized in Figure~\ref{figure:clones}.

\begin{table*} \scriptsize
\caption{\label{table:clones}%
      List of all Boolean clones with definitions and bases, where
      $\id(x) = x$ and $h_n(x_1, \ldots, x_{n+1}) =
      \bigvee^{n+1}_{i=1}x_1 \cdots x_{i-1}x_{i+1} \cdots x_{n+1}$,
      $\dual(f)(a_1, \ldots, a_n) = 1 - f(\overbar{a_1}, \ldots,
      \overbar{a_n})$.
%where $t^q_p$ denotes the $q$-ary $p$-threshold function.
      }
\begin{tabularx}{\textwidth}{l l l}
  \hline
      Clone & Definition & Base \\
      \hline
      $\CloneBF$ & All Boolean functions & $\{x \land y, \neg x\}$ \\
      $\CloneR_0$ & $\{f \mid f \text{ is $0$-reproducing}\}$ & $\{x \land y, x \oplus y\}$ \\
      $\CloneR_1$ & $\{f \mid f \text{ is $1$-reproducing}\}$ & $\{x \lor y, x \oplus y \oplus 1 \}$ \\
      $\CloneR_2$ & $\CloneR_0 \cap \CloneR_1$ & $\{x \lor y, x \land (y \oplus z \oplus 1) \}$ \\
      $\CloneM$ & $\{f \mid f \text{ is monotonic}\}$ & $\{x \lor y, x \land y, 0, 1\}$ \\
      $\CloneM_1$ & $\CloneM \cap \CloneR_1$ & $\{x \lor y, x \land y, 1\}$ \\
      $\CloneM_0$ & $\CloneM \cap \CloneR_0$ & $\{x \lor y, x \land y, 0\}$ \\
      $\CloneM_2$ & $\CloneM \cap \CloneR_2$ & $\{x \lor y, x \land y \}$ \\
%      $\CloneS^n_{0}$ & $\{f \mid f \text{ is $0$-separating of
%      degree } n\}$ & $\{x \to y, t_2^{n+1}\}$ \\
      $\CloneS^n_{0}$ & $\{f \mid f \text{ is $0$-separating of degree } n\}$ & $\{x \to y, \dual(h_n)\}$ \\
      $\CloneS_0$ & $\{f \mid f \text{ is $0$-separating}\}$ & $\{x \to y\}$ \\
%      $\CloneS^n_{1}$ & $\{f \mid f \text{ is $1$-separating of degree } n\}$ & $\{x \wedge \neg y, t_n^{n+1}\}$ \\
      $\CloneS^n_{1}$ & $\{f \mid f \text{ is $1$-separating of degree } n\}$ & $\{x \wedge \neg y, h_n\}$ \\
      $\CloneS_1$ & $\{f \mid f \text{ is $1$-separating}\}$ & $\{x \wedge \neg y\}$ \\
%      $\CloneS^n_{02}$ & $\CloneS^n_0 \cap \CloneR_2$ & $\{x \lor  (y  \land  \neg z), t_2^{n+1}\}$ \\
      $\CloneS^n_{02}$ & $\CloneS^n_0 \cap \CloneR_2$ & $\{x \lor  (y  \land  \neg z), \dual(h_n)\}$ \\
      $\CloneS_{02}$ & $\CloneS_0 \cap \CloneR_2$ & $\{x \lor  (y  \land  \neg z)\}$ \\
%      $\CloneS^n_{01}$ & $\CloneS^n_0 \cap \CloneM$ & $\{t_2^{n+1},1\}$ \\
      $\CloneS^n_{01}$ & $\CloneS^n_0 \cap \CloneM$ & $\{\dual(h_n),1\}$ \\
      $\CloneS_{01}$ & $\CloneS_0 \cap \CloneM$ & $\{x \lor  (y  \land z), 1\}$ \\
%      $\CloneS^n_{00}$ & $\CloneS^n_0 \cap \CloneR_2 \cap \CloneM$ & $\{x \lor  (y  \land  z), t_2^{n+1}\}$ \\
      $\CloneS^n_{00}$ & $\CloneS^n_0 \cap \CloneR_2 \cap \CloneM$ & $\{x \lor  (y  \land  z), \dual(h_n)\}$ \\
      $\CloneS_{00}$ & $\CloneS_0 \cap \CloneR_2 \cap \CloneM$ & $\{x \lor  (y  \land  z)\}$ \\
%      $\CloneS^n_{12}$ & $\CloneS^n_1 \cap \CloneR_2$ & $\{x \land  (y  \lor  \neg z), t_n^{n+1}\}$ \\
      $\CloneS^n_{12}$ & $\CloneS^n_1 \cap \CloneR_2$ & $\{x \land  (y  \lor  \neg z), h_n\}$ \\
      $\CloneS_{12}$ & $\CloneS_1 \cap \CloneR_2$ & $\{x \land  (y  \lor  \neg z)\}$ \\
%      $\CloneS^n_{11}$ & $\CloneS^n_1 \cap \CloneM$ & $\{t_n^{n+1}, 0\}$ \\
      $\CloneS^n_{11}$ & $\CloneS^n_1 \cap \CloneM$ & $\{h_n, 0\}$ \\
      $\CloneS_{11}$ & $\CloneS_1 \cap \CloneM$ & $\{x \land  (y  \lor  z), 0\}$ \\
%      $\CloneS^n_{10}$ & $\CloneS^n_1 \cap \CloneR_2 \cap \CloneM$ & $\{x \land  (y  \lor  z), t_n^{n+1}\}$ \\
      $\CloneS^n_{10}$ & $\CloneS^n_1 \cap \CloneR_2 \cap \CloneM$ & $\{x \land  (y  \lor  z), h_n\}$ \\
      $\CloneS_{10}$ & $\CloneS_1 \cap \CloneR_2 \cap \CloneM$ & $\{x \land  (y  \lor  z)\}$ \\
      $\CloneD$ & $\{f \mid f \text{ is self-dual}\}$ & $ \{(x \land \neg y) \lor (x \land \neg z) \lor (\neg y \land  \neg z)\}$ \\
      $\CloneD_1$ & $\CloneD \cap \CloneR_2$ & $ \{(x \land y) \lor (x \land \neg z) \lor (y \land \neg z)\}$ \\
%      $\CloneD_2$ & $\CloneD \cap \CloneM$ & $ \{t_2^3\}$ \\ % (x \land y) \lor (x \land z) \lor (y \land z)
      $\CloneD_2$ & $\CloneD \cap \CloneM$ & $ \{h_2\}$ \\ % (x \land y) \lor (x \land z) \lor (y \land z)
      $\CloneL$ & $\{f \mid f \text{ is affine}\}$ & $\{ x \oplus y,1\}$ \\
      $\CloneL_0$ & $\CloneL \cap \CloneR_0$ & $\{x \oplus y \}$ \\
      $\CloneL_1$ & $\CloneL \cap \CloneR_1$ & $\{x \oplus y  \oplus 1 \}$ \\
      $\CloneL_2$ & $\CloneL \cap \CloneR_2$ & $\{x \oplus y \oplus z\}$ \\
      $\CloneL_3$ & $\CloneL \cap \CloneD$ & $\{x  \oplus  y  \oplus  z  \oplus  1\}$ \\
      $\CloneV$ & $\{f \mid  f $ is a disjunction or constants$\}$ & $\{ x \lor y, 0,1 \}$ \\
      $\CloneV_0$ & $\CloneV \cap \CloneR_0$ & $\{ x \lor y, 0\}$ \\
      $\CloneV_1$ & $\CloneV \cap \CloneR_1$ & $\{ x \lor y, 1\}$ \\
      $\CloneV_2$ & $\CloneV \cap \CloneR_2$ & $\{ x \lor y\}$ \\
      $\CloneE$ & $\{f \mid  f $ is a conjunction or constants$\}$ & $\{ x \land y, 0, 1 \}$ \\
      $\CloneE_0$ & $\CloneE \cap \CloneR_0$ & $\{ x \land y, 0\}$ \\
      $\CloneE_1$ & $\CloneE \cap \CloneR_1$ & $\{ x \land y, 1\}$ \\
      $\CloneE_2$ & $\CloneE \cap \CloneR_2$ & $\{ x \land y\}$ \\
      $\CloneN$ & $\{f \mid f $ depends on at most one variable$\}$ & $\{ \neg x,0,1\}$ \\
      $\CloneN_2$ & $\CloneN \cap \CloneR_2$ & $\{ \neg x\}$ \\
      $\CloneI$ & $\{f \mid f \text{ is a projection or a constant}\}$ & $\{\id, 0,1\}$ \\
      $\CloneI_0$ & $\CloneI \cap \CloneR_0$ & $\{\id, 0\}$ \\
      $\CloneI_1$ & $\CloneI \cap \CloneR_1$ & $\{\id, 1\}$ \\
      $\CloneI_2$ & $\CloneI \cap \CloneR_2$ & $\{\id\}$ \\
      \hline
    \end{tabularx}
  \end{table*}

\begin{figure}
\centering
\includegraphics[scale=0.6]{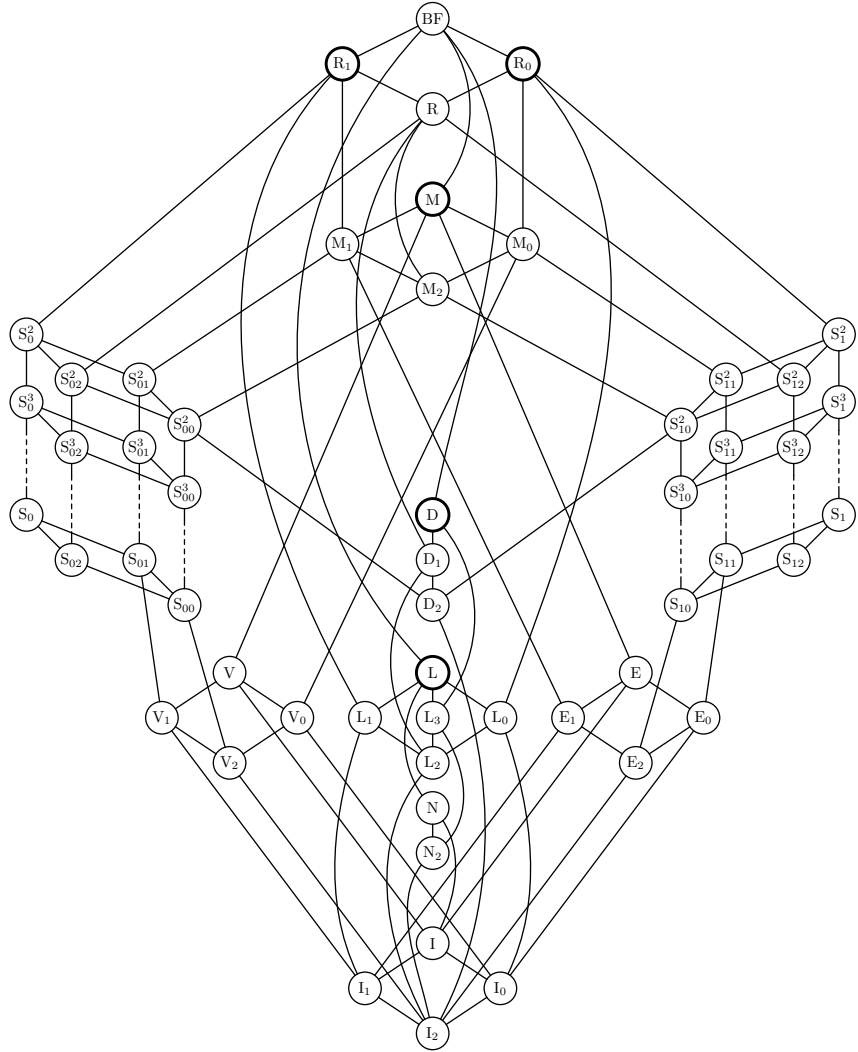}
%\vspace{100pc}
\caption{The lattice of Boolean clones.}
\label{figure:clones}
\end{figure}

\subsection{Weak Bases} \label{appendix:weak_bases}
We extend the definition of a polymorphism and say that a partial
function $f$ is a {\em partial polymorphism} to a relation $R$ if $R$
is closed under $f$ for every sequence of tuples for which $f$ is
defined. A set of partial functions $\cc{F}{}{}$ is said to be a {\em
  strong partial clone} if it contains all (total and partial)
projection functions and is closed under composition of functions.  By
$\ppol(\Gamma)$ we denote the set of partial polymorphisms to the set
of relations $\Gamma$. Obviously sets of the form $\ppol(\Gamma)$
always form strong partial clones and again we have a Galois
connection between clones and co-clones.

%BKKR69i,BKKR69ii,
\begin{theorem}~\cite{romov1981}
  Let $\Gamma$ and $\Gamma'$ be two sets of relations. Then
  $\pcclone{\Gamma} \subseteq \pcclone{\Gamma'}$ if and only if 
  $\ppol(\Gamma') \subseteq \ppol(\Gamma)$.
\end{theorem}

We define the {\em weak base} of a co-clone $\cc{IC}{}{}$ to be the
base of the smallest member of the interval $\clone{I}(\cc{IC}{}{}) =
\{\cc{ID}{}{} \mid \cc{ID}{}{} = \pcclone{\cc{ID}{}{}}$ and
$\cclone{\cc{ID}{}{}} = \cc{IC}{}{}\}$. Weak bases were first
introduced in Schnoor and Schnoor~\cite{schnoor2008a,schnoor2008b} but
their construction resulted in relations that were in many cases
exponentially larger than the plain bases with respect to arity. Weak
bases fulfilling additional minimality conditions was given in
Lagerkvist~\cite{Lagerkvist2014} using relational descriptions. By
construction the weak base of a co-clone is always a single relation.

\begin{theorem} [\cite{schnoor2008a}] \label{theorem:weak_base}
  Let $R_w$ be the weak base of some co-clone $\cc{IC}{}{}$. Then for any
  finite base $\Gamma$ of $\cc{IC}{}{}$ it holds that $R_w \in
  \pcclone{\Gamma}{}{}$.
\end{theorem}

See Table~\ref{table:weak_bases} for a complete list of weak bases.
\begin{table*}  \scriptsize
\caption{Weak bases for all Boolean co-clones with a finite base}
\label{table:weak_bases}
\begin{tabularx}{\textwidth}{l l}
  \hline
  Co-clone & Weak base \\
  \hline
  $\cc{IBF}{}{}$ &  Eq$(x_1,x_2)$ \\
  $\cc{IR}{}{0}$ &  $\F(c_0)$  \\
  $\cc{IR}{}{1}$ &  $\T(c_1)$  \\
  $\cc{IR}{}{2}$ &  $\F(c_0) \wedge \T(c_1)$  \\

  $\cc{IM}{}{}$ &  $(x_1 \rightarrow x_2)$  \\
  $\cc{IM}{}{0}$  & $(x_1 \rightarrow x_2) \wedge \F(c_0)$ \\ 

  $\cc{IM}{}{1}$  & $(x_1 \rightarrow x_2) \wedge \T(c_1)$ \\

  $\cc{IM}{}{2}$ & $(x_1 \rightarrow x_2) \wedge \F(c_0)
  \wedge \T(c_1)$  \\

  $\cc{IS}{n}{0}, n \geq 2 $ & $\orn{}{n}(x_1, \ldots, x_n) \wedge \T(c_1)$  \\
 % $\cc{IS}{}{0} $ & $\infty$   \\

%  $\cc{IS}{}{1} $ & $\infty$   \\

  $\cc{IS}{n}{02}, n \geq 2$ & $\orn{}{n}(x_1, \ldots, x_n) \wedge
  \F(c_0) \wedge \T(c_1)$ \\
  $\cc{IS}{n}{01}, n \geq 2$ & $\orn{}{n}(x_1, \ldots, x_n)
  \wedge (x \rightarrow x_1 \cdots x_n)
   \wedge \T(c_1)$  \\
  $\cc{IS}{n}{00}, n \geq 2$ & $\orn{}{n}(x_1, \ldots, x_n)
  \wedge (x \rightarrow x_1 \cdots x_n)
    \wedge \F(c_0) \wedge \T(c_1)$  \\
%  $\cc{IS}{}{02} $ & $\infty$   \\
  $\cc{IS}{n}{1}, n \geq 2$ & $\nandn{}{n}(x_1, \ldots, x_n) \wedge
  \F(c_0)$  \\
  $\cc{IS}{n}{12}, n \geq 2 $ & $\nandn{}{n}(x_1, \ldots, x_n) \wedge
  \F(c_0) \wedge \T(c_1)$  \\
%  $\cc{IS}{}{12} $ & $\infty$   \\
%  $\cc{IS}{}{01} $ & $\infty$   \\
  $\cc{IS}{n}{11}, n \geq 2$ & $\nandn{}{n}(x_1, \ldots, x_n)
  \wedge (x \rightarrow x_1 \cdots x_n)
   \wedge \F(c_0)$  \\
%  $\cc{IS}{}{11} $ & $\infty$   \\

%  $\cc{IS}{n}{01} $    \\
%  $\cc{IS}{}{01} $ & $\infty$   \\
%  $\cc{IS}{n}{11} $    \\
%  $\cc{IS}{}{11} $ & $\infty$   \\

%  $\cc{IS}{}{00} $ & $\infty$   \\
  $\cc{IS}{n}{10}, n \geq 2$ & $\nandn{}{n}(x_1, \ldots, x_n)
  \wedge (x \rightarrow x_1 \cdots x_n)
   \wedge \F(c_0) \wedge \T(c_1)$  \\
%  $\cc{IS}{}{10} $ & $\infty$   \\

  $\cc{ID}{}{}$ &  $(x_1 \neq x_2)$ \\
  $\cc{ID}{}{1}$  & $(x_1 \neq x_2) \wedge \F(c_0) \wedge \T(c_1)$  \\
  $\cc{ID}{}{2}$  & $\orn{2 \neq}{2}(x_1,x_2,x_3,x_4) \wedge \F(c_0) \wedge \T(c_1)$  \\

  $\cc{IL}{}{}$   & $\evenn{}{4}(x_1,x_2,x_3,x_4)$ \\
  $\cc{IL}{}{0}$  & $\evenn{}{3}(x_1,x_2,x_3) \wedge \F(c_0)$  \\
  $\cc{IL}{}{1}$  & $\oddn{}{3}(x_1,x_2,x_3) \wedge \T(c_1)$  \\
  $\cc{IL}{}{2}$  & $\evenn{3 \neq}{3}(x_1,\ldots,x_6) \wedge \F(c_0) \wedge \T(c_1)$  \\
  $\cc{IL}{}{3}$  & $\evenn{4 \neq}{4}(x_1,\ldots,x_8)$  \\

  $\cc{IV}{}{}$  & $(\overbar{x_1} \leftrightarrow
\overbar{x_2}\overbar{x_3}) \wedge (\overbar{x_2} \vee \overbar{x_3}
\rightarrow \overbar{x_4})$ \\
  $\cc{IV}{}{0}$  & $(\overbar{x_1} \leftrightarrow \overbar{x_2}\overbar{x_3}) \wedge \F(c_0)$  \\
  $\cc{IV}{}{1}$  & $(\overbar{x_1} \leftrightarrow
\overbar{x_2}\overbar{x_3}) \wedge (\overbar{x_2} \vee \overbar{x_3}
\rightarrow \overbar{x_4}) \wedge \T(c_1)$  \\
  $\cc{IV}{}{2}$  & $(\overbar{x_1} \leftrightarrow
\overbar{x_2}\overbar{x_3}) \wedge \F(c_0) \wedge \T(c_1)$  \\

  $\cc{IE}{}{}$  & $(x_1 \leftrightarrow x_2x_3) \wedge (x_2 \vee x_3
\rightarrow x_4)$ \\
  $\cc{IE}{}{0}$  & $(x_1 \leftrightarrow x_2x_3) \wedge (x_2 \vee x_3
\rightarrow x_4) \wedge \F(c_0)$  \\
  $\cc{IE}{}{1}$  & $(x_1 \leftrightarrow x_2x_3) \wedge \T(c_1)$  \\
  $\cc{IE}{}{2}$  & $(x_1 \leftrightarrow x_2x_3) \wedge \F(c_0) \wedge \T(c_1)$  \\

  $\cc{IN}{}{}$  & $\evenn{}{4}(x_1,x_2,x_3,x_4) \wedge x_1x_4
\leftrightarrow x_2x_3$ \\
  $\cc{IN}{}{2}$  & $\evenn{4 \neq}{4}(x_1,\ldots,x_8) \wedge x_1x_4
\leftrightarrow x_2x_3$  \\  

  $\cc{II}{}{}$  & $(x_1 \leftrightarrow x_2x_3) \wedge (\overbar{x_4}
\leftrightarrow \overbar{x_2}\overbar{x_3})$ \\
  $\cc{II}{}{0}$  & $(\overbar{x_1} \vee \overbar{x_2}) \wedge (\overbar{x_1}\overbar{x_2} \leftrightarrow \overbar{x_3}) \wedge \F(c_0)$  \\
  $\cc{II}{}{1}$  & $(x_1 \vee x_2) \wedge (x_1x_2 \leftrightarrow x_3) \wedge \T(c_1)$  \\
  $\cc{II}{}{2}$  & $\Rddd(x_1,\ldots,x_6) \wedge \F(c_0)
\wedge \T(c_1)$ \\
  \hline
\end{tabularx}
\end{table*}

\subsection{Additional Proofs for Section~\ref{section:eth}} \label{appendix:eth}

%\subsubsection* {CSP$(\{R_2\})-2$ to (Unweighted) Max-Ones$(\{R_4\})$}
\begin{lemma} \label{mored1}
  $\SAT(\Rel{II_2})$-2 LV-reduces to $\MO(\Rel{IL_2})$.
\end{lemma}
\begin{proof}
We reduce an instance $I$ of $\SAT(\Rel{II_2})$-$2$ on $n$ variables %and $m$
constraints to an instance of $\MO(\Rel{IL_2})$ containing at
most $2+8n$ variables.
Let $v_0,v_1$ be two fresh global variables constrained as
$\Rel{IL_2}(v_0,v_0,v_0,v_1,v_1,v_1,v_0,v_1)$. Note that this forces
$v_0$ to $0$ and $v_1$ to $1$ in any satisfying assignment. Now, for
every variable $x$ in the $\SAT$-instance we create an additional variable $x'$ which we 
constrain as $\Rel{IL_2}(x',x,v_1, x,x',v_0,v_0,v_1)$.  This
correctly implements $\neqq(x,x')$.
%forces
%$x$ and $x'$ to be mapped to different values.
%
For the $i$-th constraint, $\Rel{II_2}(x_1,\dots,x_6,c_0,c_1)$, in $I$ we
create three variables $z_i^1, z_i^2, z_i^3$ and constrain them
as $\Rel{IL_2}(z_i^1,z_i^2,z_i^3,x_1,x_2,x_3,c_0,c_1)$,
we also add the constraint $\Rel{IL_2}(x_4,x_5,x_6,x_1,x_2,x_3,c_0,c_1)$.
Since every variable in the $\SAT$-instance $I$ can occur in at most two
constraints we have that $m \leq 2n$.  Hence the resulting $\MO$
instance contains at most $2+2n+3\cdot 2n = 2 + 8n$ variables.
Since $x$ and $x'$, and $v_0$ and $v_1$, must take different values it holds that the measure of a solution of this new instance is exactly the number of variables $z_i^j$ that are mapped to $1$.
Hence, for an optimal solution the objective value is $\ge 2m$ if and only if $I$ is satisfiable.
\qed
\end{proof}

%\subsubsection* {(Unweighted) Max-Ones$(\{R_4\})$ to (Unweighted) Max-Ones$(\{P_3\})$}
\begin{lemma} \label{mored2}
  $\MO(\Rel{IL_2})$ LV-reduces to $\MO(\Rel{IL_0})$.
\end{lemma}
\begin{proof}
We reduce an instance $I$ of $\MO(\Rel{IL_2})$ on $n$ variables 
to an instance of $\MO(\Rel{IL_0})$ on $2+2n$ variables.
Let $v_0, v_1, y_1,\dots,y_n$ be fresh variables and constrain them as
$\Rel{IL_0}(v_0,v_0,v_0,v_0) \wedge \Rel{IL_0}(v_1,v_0,y_1,v_0) \wedge \ldots \wedge
\Rel{IL_0}(v_1,v_0,y_n,v_0)$.  Note that this forces $v_0$ to $0$, and
that if $v_1$
is mapped to $0$, then so are the variables $y_1,\dots,y_n$. If $v_1$
is mapped to $1$ on the other hand, then $y_1,\dots,y_n$ can be
mapped to $1$.
For every constraint $\Rel{IL_2}(x_1,x_2,x_3,x_4,x_5,x_6,c_0,c_1)$ we create the
constraints
$\Rel{IL_0}(x_1,x_2,x_3,v_0) \wedge
 \Rel{IL_0}(v_1,x_1,x_4,v_0) \wedge
 \Rel{IL_0}(v_1,x_2,x_5,v_0) \wedge
 \Rel{IL_0}(v_1,x_3,x_6,v_0) \wedge
 \Rel{IL_0}(v_1,c_0,c_1,v_0)$.
The resulting $\MO(\Rel{IL_0})$ instance has $2 + 2n$
variables and has a solution with measure $n+1+k$ if and only if $I$ has a solution with measure $k$.
\qed
\end{proof}

%\subsubsection* {(Unweighted) Max-Ones$(\{R_2\})$ to (Unweighted) Max-Ones$(\{R_7\})$}
\begin{lemma} \label{mored3}
  $\MO(\Rel{II_2})$ LV-reduces to $\MO(\Rel{IN_2})$.
\end{lemma}

\begin{proof}
We reduce an instance $I$ of $\MO(\Rel{II_2})$ over $n$
variables to an instance of $\MO(\Rel{IN_2})$ over $2 + 3n$
variables.
Create two fresh variables $v_0,v_1$ and constrain them as
$\Rel{IN_2}(v_0,v_0,v_0,v_0,v_1,v_1,v_1,v_1)$ in order to force
$v_0$ and $v_1$ to be mapped to different values.
We then create the $2n$
variables $y_1,\ldots,y_{2n}$ and constrain them as $\bigwedge_{i=1}^{2n}
\Rel{IN_2}(v_0,v_0,v_0,v_0,y_i,y_i,y_i,y_i)$.
This forces all of the
variables $y_i$ to be mapped to the same value as $v_1$.
We can now 
%implement
express
$\Rel{II_2}(x_1, x_2, x_3, x_4, x_5, x_6, c_0, c_1)$ using
the implementation
$\Rel{IN_2}(v_0,x_1,x_2,x_6,v_1,x_4,x_5,x_3) \wedge
\Rel{IN_2}(v_0,c_0,c_0,v_0,v_1,c_1,c_1,v_1)$.
Note that in any optimal solution of the new instance $v_1$ will be mapped to
$1$ which means that the implementation of $\Rel{II_2}$ given above
will be correct.
The resulting instance has a solution with measure $1+2n+k$ if and only if $I$ has a solution with measure $k$.
\qed
\end{proof}

%\subsubsection* {(Unweighted) Max-Ones$(\{R_{\le 1}\})$ to
%  (Unweighted) Max-Ones$(\{R_1\})$}
\begin{lemma} \label{mored4}
  $\MO(\Rel{IS^{2}_{1}})$ LV-reduces to $\MO(\Rel{ID_2})$.
\end{lemma}
\begin{proof}
We reduce an instance of $\MO(\Rel{IS^{2}_{1}})$ on $n$ variables
to an instance of $\MO(\Rel{ID_2})$ on $2+3n$ variables.
Create two new variables $v_0$ and $v_1$ and constrain them as
$\Rel{ID_2}(v_1,v_1,v_0,v_0,v_0,v_1)$.  Note that this forces $v_0$ to
$0$ and $v_1$ to $1$. For every variable $x$ we introduce two extra
variables $x'$ and $x''$ and constrain them as
$\Rel{ID_2}(x,x',x',x,v_0,v_1) \wedge
\Rel{ID_2}(x',x'',x'',x',v_0,v_1)$. Note that this implements the
constraints $\neqq(x,x')$ and $\neqq(x',x'')$, and that no matter what $x$ is mapped
to exactly one of $x'$ and $x''$ is mapped to $1$. For every
constraint $\Rel{IS^{2}_{1}}(x,y,c_0)$ we then introduce the
constraint $\Rel{ID_2}(x',y',x,y,c_0,v_1)$.
%and $x''$ and constrain them as $\Rel{ID_2}(v_0,x,v_1,x',v_0,v_1) \wedge
%\Rel{ID_2}(v_0,x',v_1,x'')$. Note that this implements the constraints $x
%
The resulting instance has a solution with measure $1+n+k$ if and only if $I$ has a solution with measure $k$.
\qed
\end{proof}

\begin{lemma} \label{mored5}
  $\MO(\Rel{IL_2})$ LV-reduces to $\MO(\Rel{IL_3})$.
\end{lemma}
\begin{proof}
We reduce an instance of $\MO(\Rel{IL_2})$ on $n$ variables 
%and $m$ constraints
to an instance of $\MO(\Rel{IL_3})$ on $2+3n$ variables.
Create two new variables $v_0$ and $v_1$ and constrain them as
$R_{IL_3}(v_0,v_0,v_0,v_0,v_1,v_1,v_1,v_1)$.
Note that this forces $v_0$ and $v_1$ to be mapped to different values.
We then introduce fresh variables $y_1,\dots,y_{2n}$ and constrain
them as $\bigwedge_{i=1}^{2n} R_{IL_3}(v_0,v_0,v_0,v_0,y_i,y_i,y_i,y_i)$.  This
will ensure that every variables $y_i$ is mapped to the same value as $v_1$ and therefore that in every optimal solution $v_0$ is mapped to $0$ and
$v_1$ is mapped to $1$.
For every constraint $R_{IL_2}(x_1,\dots,x_6,c_0,c_1)$ we introduce the
constraints
$R_{IL_3}(c_0,x_1,x_2,x_3,c_1,x_4,x_5,x_6) \wedge 
 R_{IL_3}(c_0,c_0,c_0,c_0,v_1,v_1,v_1,v_1) \wedge
 R_{IL_3}(v_0,v_0,v_0,v_0,c_1,c_1,c_1,c_1)$.
The resulting instance has a solution with measure $1+2n+k$ if and only if $I$ has a solution with measure $k$.
\qed
\end{proof}

 \begin{lemma} \label{lemma:ethmaxones}
 If the ETH is false, then
 $\MO(\Gamma) \in \SE$ for every
 finite Boolean constraint language $\Gamma$.
 \end{lemma}
 \begin{proof}
Define SNP to be the class of properties expressible by formulas
of the type
$\exists S_1\ldots\exists S_n\forall x_1\ldots\forall x_m . F$
where
$F$ is a quantifier-free logical formula,
$\exists S_1\ldots\exists S_n$ are second order existential
quantifiers, and
$\forall x_1\ldots\forall x_m$ are first-order universal
quantifiers.
%~\cite{papadimitriou1991}. 
Monadic SNP (MSNP) is the restriction of SNP where all second-order
predicates are required to be unary~\cite{FeVa98}. The associated
search problem tries to identify instantiations of $S_1,\ldots,S_n$
that make the resulting first-order formula true.  We will be
interested in properties that can be expressed by formulas that
additionally contain {\em size-constrained} existential quantifier.  A
size-constrained existential quantifier is of the form $\exists S$,
$|S| \oplus s$, where $|S|$ is the number of inputs where relation $S$
holds, and $\oplus \in \{=, \leq , \geq \}$.  Define size-constrained
SNP as the class of properties of relations and numbers that are
expressible by formulas $\exists S_1\ldots\exists S_n\forall
x_1\ldots\forall x_m . F$ where the existential quantifiers are
allowed to be size-constrained.

If the ETH is false then $3$-$\SAT$ is solvable in subexponential time. By
Impagliazzo et al.$\,$~\cite{impagliazzo2001} this problem is
size-constrained MSNP-complete under size-preserving SERF
reductions. Hence we only have to prove that $\MO(\cdot)$ is included
in size-constrained MSNP for it to be solvable in subexponential
time. Impagliazzo et al.$\,$~\cite{impagliazzo2001} shows that $k$-SAT is
in SNP by providing an explicit formula $\exists S . F$ where $F$ is a
universal formula and $S$ a unary predicate interpreted such that $x
\in S$ if and only if $x$ is true. Let $k$ be the highest arity of any
relation in $\Gamma$. Since $k$-$\SAT$ can q.p.p.\ implement any $k$-ary
relation it is therefore sufficient to prove that $\MO(\Gamma^k_{\rm
  SAT})$ is in size-constrained MSNP, where $\Gamma^k_{\rm SAT}$ is
the language corresponding to all satisfying assignments of
$k$-SAT. This is easy to do with the formula

\[\exists S, |S| \geq K . F\]

where $K$ is the parameter corresponding to the number of variables
that has to be assigned 1.
\qed
\end{proof}

\begin{lemma} \label{lemma:ethvcsp}
If $\MO(\Gamma) \in \SE$ for every
finite Boolean constraint language $\Gamma$, then
$\UWVCSP_{d}(\Delta) \in \SE$ for every
finite set of Boolean cost functions $\Delta$ and arbitrary $d \geq 0$.
\end{lemma}
\begin{proof}
We first show that if every $\MO(\Gamma) \in \SE$, then
the minimization variant $\MinO(\Gamma) \in \SE$ for all $\Gamma$, too.
Arbitrarily choose a finite constraint language $\Gamma$ over $\B$.
We present an LV-reduction from $\MinO(\Gamma)$ to $\MO(\Gamma \cup \{\neqq\})$.
Let $(\{v_1,\ldots,v_n\},C)$ be an arbitrary instance of $\MinO(\Gamma)$ with optimal value $K$.
Consider the following instance $I'$ of $\MO(\Gamma \cup \{\neqq\})$:
\[(\{v_1,v_1',v_1'',\ldots,v_n,v_n',v_n''\},C \cup \{\neqq(v_1,v_1'),\neqq(v_1,v_1''),\ldots,\neqq(v_n,v_n'),\neqq(v_n,v_n''\}).\]
For each variable $v_i \in \{v_1,\ldots,v_n\}$ that is assigned 0, the corresponding
variables $v_i',v_i''$ are assigned 1, and vice-versa.
It follows that the optimal value of $I'$ is $2n-K$.
Hence, $\MinO(\Gamma) \in \SE$
since $\MO(\Gamma \cup \{\neqq\}) \in \SE$.

\medskip

Now, arbitrarily choose $d \geq 0$ and a finite set of
Boolean cost functions $\Delta$.
Since $\Delta$ is finite, we may
without loss of generality assume that each function $f \in \Delta$ has its range in $\{0,1,2,\ldots\}$.

We show that $\UWVCSP_{d}(\Delta) \in \SE$ by exhibiting
an LV-reduction from $\UWVCSP_{d}(\Delta)$ to $\MinO(\Gamma)$ where $\Gamma$ is finite
and only depends on $\Delta$.
Given a tuple $\tup{a} = (a_1,\ldots,a_k) \in \B^k$, let
$\val(\tup{a}) = 1+\sum_{j : a_j=1} 2^{j-1}.$
For each $f \in \Delta$ of arity $k$, define
\begin{align*}
R_f &= 
\left\{(x_1,\ldots,x_k,y_1,\ldots,y_{2^k}) \in \B^{k+2^k} \middle|
\begin{aligned}
&f(x_1,\ldots,x_k) > 0,\\
&\{ i : y_i \ne 0\}=\{\val(x_1,\ldots,x_k)\}
\end{aligned}
\right\}\\
& \quad\cup
\{(x_1,\ldots,x_k,0,\ldots,0) \in \B^{k+2^k} \mid f(x_1,\ldots,x_k) =0\},
\end{align*}
%\[R_f = \{(x_1,\ldots,x_k,y_1,\ldots,y_{2^k+1}) \in \B^{k+2^k} \mid\]
%\[ \; \mbox{if $f(x_1,\ldots,x_k) > 0$, then
%$y_t=1$ where $t=\val(x_1,\ldots,x_k)$ and $y_i=0$ when $i \neq t$}\} \cup\]
%\[\{(x_1,\ldots,x_k,0,\ldots,0) \in \B^{k+2^k} \mid \mbox{if $f(x_1,\ldots,x_k) =0$}\},\]
and let $\Gamma=\{\eq,\neqq\} \cup \{R_f \mid f \in \Delta\}$.

One may interpret $R_f$ as follows: for each $(x_1,\ldots,x_k) \in \B^k$
%,
the relation
 $R_f$ contains exactly
one tuple $(x_1,\ldots,x_k,y_1,\ldots,y_{2^k})$. If $f(x_1,\ldots,x_k)=0$,
then this is the tuple $(x_1,\ldots,x_k,0,\ldots,0)$. If $f(x_1,\ldots,x_k) > 0$,
then this is the tuple $(x_1,\ldots,x_k,0,\ldots,1,\ldots,0)$ where the 1 is in position
$k+\val(x_1,\ldots,x_k)$. We show below how $R_f$ can be used for ``translating''
each $\tup{x} \in \B^k$ into its corresponding weight as prescribed by $f$.
%We claim that there exists an LV-reduction from
%$\UWVCSP_{d}(\Delta)$ to
%$\MinO(\Gamma)$.

Let $(V,\sum_{i=1}^m f_i(\tup{x}_i))$ be an arbitrary instance of
$\UWVCSP_{d}(\Delta)$ where $V=\{v_1,\ldots,v_n\}$.
Let $\arity(f_i)$ denote the arity of function $f_i$.
Assume the instance has an optimal solution with value $K$.
For each term $f_i(v_1,\ldots,v_k)$ in the sum, do the following:

\begin{enumerate}
\item
introduce $2^{k}$ fresh variables $v'_{1},\ldots,v'_{2^k}$,

\item
introduce $k$ fresh variables $w_{1},\ldots,w_{k}$,

\item
for each $\tup{a} \in \B^k$ such that $f(\tup{a}) > 1$,
introduce $n'=f(\tup{a})$ fresh variables $u_0,\ldots,u_{n'-1}$,

\item
introduce the constraint
$R_f(v_{1},\ldots,v_{k},v'_{1},\ldots,v'_{2^k})$,

\item
introduce the constraints
$\neqq(v_{1},w_{1}),\ldots,\neqq(v_{k},w_{k})$, and

\item
for each $\tup{a} \in \B^k$, let $n'=f(\tup{a})$ and do the following
if $n' > 1$:   let $p=\val(\tup{a})$ and introduce
the constraints $\eq(v'_p,u_0),\eq(u_0,u_1),\ldots,\eq(u_{n'-2},u_{n'-1})$.

\end{enumerate}

It is not difficult to realize that the resulting instance has
optimal value $K+\sum_{i=1}^m \arity(C_i)$ given the interpretation of $R_f$
and the following motivation of step 5: the $\neqq$ constraints introduced in step 5
ensure that the weight of $(x_1,\ldots,x_k)$
does not influence the weight of the construction and this explains that
we need to adjust the optimal value with $\sum_{i=1}^m \arity(C_i)$.

Furthermore, the instance contains at most

\[|V| + |C| \cdot (2s + t \cdot (2^s+1))\]

variables
where $s=\max\{\arity(f) \mid f \in \Delta\}$ and
$t=\max\{f(\tup{a}) \mid f \in \Delta \; {\rm and} \; \tup{a} \in \B^{\arity(f)}\}$.
By noting that $|C| \leq d|V|$ and that $s,t$ are constants
that only depend on $\Delta$, it follows that the reduction
is an LV-reduction.
\qed
\end{proof}

\end{document}